\PassOptionsToPackage{unicode}{hyperref}
\PassOptionsToPackage{hyphens}{url}
\PassOptionsToPackage{dvipsnames,svgnames,x11names}{xcolor}
\documentclass[
  12pt]{article}

\usepackage{amsmath,amssymb}
\usepackage{iftex}
\ifPDFTeX
  \usepackage[T1]{fontenc}
  \usepackage[utf8]{inputenc}
  \usepackage{textcomp} 
\else 
  \usepackage{unicode-math}
  \defaultfontfeatures{Scale=MatchLowercase}
  \defaultfontfeatures[\rmfamily]{Ligatures=TeX,Scale=1}
\fi
\usepackage{lmodern}
\ifPDFTeX\else  
\fi
\IfFileExists{upquote.sty}{\usepackage{upquote}}{}
\IfFileExists{microtype.sty}{
  \usepackage[]{microtype}
  \UseMicrotypeSet[protrusion]{basicmath} 
}{}
\makeatletter
\@ifundefined{KOMAClassName}{
  \IfFileExists{parskip.sty}{%
    \usepackage{parskip}
  }{
    \setlength{\parindent}{0pt}
    \setlength{\parskip}{6pt plus 2pt minus 1pt}}
}{
  \KOMAoptions{parskip=half}}
\makeatother
\usepackage{xcolor}
\makeatletter
\ifx\paragraph\undefined\else
  \let\oldparagraph\paragraph
  \renewcommand{\paragraph}{
    \@ifstar
      \xxxParagraphStar
      \xxxParagraphNoStar
  }
  \newcommand{\xxxParagraphStar}[1]{\oldparagraph*{#1}\mbox{}}
  \newcommand{\xxxParagraphNoStar}[1]{\oldparagraph{#1}\mbox{}}
\fi
\ifx\subparagraph\undefined\else
  \let\oldsubparagraph\subparagraph
  \renewcommand{\subparagraph}{
    \@ifstar
      \xxxSubParagraphStar
      \xxxSubParagraphNoStar
  }
  \newcommand{\xxxSubParagraphStar}[1]{\oldsubparagraph*{#1}\mbox{}}
  \newcommand{\xxxSubParagraphNoStar}[1]{\oldsubparagraph{#1}\mbox{}}
\fi
\makeatother

\usepackage{longtable,booktabs,array}
\usepackage{calc} 
\usepackage{etoolbox}
\makeatletter
\patchcmd\longtable{\par}{\if@noskipsec\mbox{}\fi\par}{}{}
\makeatother
\IfFileExists{footnotehyper.sty}{\usepackage{footnotehyper}}{\usepackage{footnote}}
\makesavenoteenv{longtable}
\usepackage{graphicx}
\makeatletter
\def\maxwidth{\ifdim\Gin@nat@width>\linewidth\linewidth\else\Gin@nat@width\fi}
\def\maxheight{\ifdim\Gin@nat@height>\textheight\textheight\else\Gin@nat@height\fi}
\makeatother
\setkeys{Gin}{width=\maxwidth,height=\maxheight,keepaspectratio}
\makeatletter
\def\fps@figure{htbp}
\makeatother

\makeatletter
\@ifpackageloaded{caption}{}{\usepackage{caption}}
\AtBeginDocument{%
\ifdefined\contentsname
  \renewcommand*\contentsname{Table of contents}
\else
  \newcommand\contentsname{Table of contents}
\fi
\ifdefined\listfigurename
  \renewcommand*\listfigurename{List of Figures}
\else
  \newcommand\listfigurename{List of Figures}
\fi
\ifdefined\listtablename
  \renewcommand*\listtablename{List of Tables}
\else
  \newcommand\listtablename{List of Tables}
\fi
\ifdefined\figurename
  \renewcommand*\figurename{Figure}
\else
  \newcommand\figurename{Figure}
\fi
\ifdefined\tablename
  \renewcommand*\tablename{Table}
\else
  \newcommand\tablename{Table}
\fi
}
\@ifpackageloaded{float}{}{\usepackage{float}}
\floatstyle{ruled}
\@ifundefined{c@chapter}{\newfloat{codelisting}{h}{lop}}{\newfloat{codelisting}{h}{lop}[chapter]}
\floatname{codelisting}{Listing}

\makeatother
\makeatletter
\@ifpackageloaded{caption}{}{\usepackage{caption}}
\@ifpackageloaded{subcaption}{}{\usepackage{subcaption}}
\makeatother

\ifLuaTeX
  \usepackage{selnolig}  
\fi
\usepackage[]{natbib}
\usepackage{bookmark}

\IfFileExists{xurl.sty}{\usepackage{xurl}}{} 
\hypersetup{
  pdftitle={Title},
  pdfauthor={Author 1; Author 2},
  pdfkeywords={3 to 6 keywords, that do not appear in the title},
  colorlinks=true,
  linkcolor={blue},
  filecolor={Maroon},
  citecolor={Blue},
  urlcolor={Blue},
  pdfcreator={LaTeX via pandoc}}

\newcommand\independent{\protect\mathpalette{\protect\independenT}{\perp}}
\def\independenT#1#2{\mathrel{\rlap{$#1#2$}\mkern2mu{#1#2}}}
\def\independenT#1#2{\mathrel{\rlap{$#1#2$}\mkern2mu{#1#2}}}

\usepackage{amsthm,amsmath,amssymb,mathrsfs,dsfont,mathtools, amsfonts}

\usepackage{algorithm}
\usepackage{algpseudocode}

\newtheorem{theorem}{Theorem}

\newtheorem{proposition}{Proposition}
\newtheorem{assumption}{Assumption}

\usepackage{dirtytalk}

\newcommand{\anon}{1}

\begin{document}

\def\spacingset#1{\renewcommand{\baselinestretch}%
{#1}\small\normalsize} \spacingset{1}


\if1\anon
{
  \title{\bf Pathway-based Bayesian factor models for 'omics data}
  \author{Lorenzo Mauri* \\
    Department of Statistical Science, Duke University, Durham, NC, USA\\
    Federica Stolf*\\
    Department of Statistical Science, Duke University, Durham, NC, USA\\
     Amy H. Herring\\
    Department of Statistical Science, Duke University, Durham, NC, USA\\
     Cameron Miller\\
    Department of Medicine, Duke University School of Medicine, Durham, NC, USA\\
    and\\
     David B. Dunson\\
    Department of Statistical Science, Duke University, Durham, NC, USA\\
    {\small *Equal contribution}}
  \maketitle
} \fi

\if0\anon
{
  \bigskip
  \bigskip
  \bigskip
  \begin{center}
    {\LARGE\bf Pathway-based Bayesian factor models for 'omics data}
\end{center}
  \medskip
} \fi

\bigskip
\begin{abstract}

Interpreting RNA-sequencing data requires identifying coordinated gene expression patterns that correspond to biological pathways. Standard factor models provide useful dimension reduction but typically ignore existing pathway knowledge or incorporate it through restrictive assumptions, limiting interpretability, and reproducibility. Here, we develop Bayesian Analysis with gene-Sets Informed Latent space (BASIL), a scalable framework for analyzing transcriptomic data that integrates annotated gene sets into latent variable inference.  BASIL places structured priors on factor loadings, shrinking them toward combinations of annotated gene sets, enhancing biological interpretability and stability, while simultaneously learning new unstructured components. BASIL provides accurate covariance estimates and uncertainty quantification, without resorting to computationally expensive Markov chain Monte Carlo sampling, by exploiting a pre-training approach that pre-estimates the latent factors. An automatic empirical Bayes procedure eliminates the need for manual hyperparameter tuning, promoting reproducibility and usability in practice. Applying BASIL to the global fever transcriptomic cohort uncovers interpretable host-response modules, with phosphoinositide signaling and interferon-driven inflammation emerging as key drivers of gene-expression variability.

\end{abstract}

\noindent%
\vfill

\newpage
\spacingset{1.8} 

\section{Introduction}

High-throughput transcriptomic technologies generate measurements for tens of thousands of genes, creating data with immense promise, but also fundamental challenges in statistical analysis, especially when the number of samples is not very large. One central task in the analysis of gene expression from high-throughput RNA-sequencing data is to extract relevant biological insights. Patterns of correlated expression among genes are frequently observed in such datasets and can  
reflect coordinated transcriptional regulation, in which genes involved in the same pathway or cellular function are regulated in concert to maintain synchronized activity \citep{myers2006finding}. 
Accurate discovery of biologically relevant signals in massive dimensional noisy data is crucial for understanding the molecular mechanisms driving differences across conditions, cell types, or perturbations. 

Matrix factorization techniques, such as principal component analysis (PCA) \citep{jolliffe2002pca} or factor analysis \citep{lawley1971factor, west03}, are widely used to infer components of variation in multivariate gene expression data.
These methods rely on latent variable representations that can effectively reduce the dimensionality of the data. When the sample size is limited and model interpretability is crucial, incorporating prior knowledge into the factorization in the form of structured data can be particularly beneficial. 
Structured matrix factorizations have been proposed to incorporate such biological information directly into the modeling framework. Examples include leveraging gene–gene interaction networks to guide factors towards known pathways \citep{elyanow2020netnmf} and using gene perturbation data to associate factors with perturbation effects in single-cell CRISPR screens \citep{zhou2023new}. Other contributions target scRNA-seq data \citep{buettner2017f, kunes2024supervised}.

In this work, we focus on bulk RNA-sequencing data, which measure genome-wide gene expression by quantifying the aggregate abundance of RNA transcripts across all cells within a sample. Here prior biological knowledge typically takes the form of gene sets. Let $Y$ denotes the normalized gene expression matrix obtained from RNA-seq or microarray experiments, with rows representing samples and columns representing genes. The prior biological information is encoded in a binary matrix $C$, where each column corresponds to a molecular pathway, gene ontology category, or any other predefined gene collection of interest. We are specifically motivated by a transcriptomic dataset from a global fever study \citep{ko2023host}, comprising $291$ patients and $14,386$ genes. As is common in such settings, the number of genes far exceeds the number of samples, making the incorporation of prior biological information particularly valuable for improving interpretability and estimation accuracy. In this cohort, whole-blood gene expression reflects coordinated host-response programs to diverse infectious etiologies and non-infectious inflammatory mimics, spanning multiple regions. 
This makes it a natural setting for pathway-informed factor modeling, where curated gene sets guide the extraction of interpretable latent modules
that summarize high-dimensional variation. 

The state of the art in this setting is PLIER (pathway-level information extractor) \citep{mao2019pathway}, a matrix factorization approach that maps gene expression data into a low-dimensional space. 
Integrating pathway-informed annotations improves interpretability by yielding latent components aligned with known pathways. The PLIER framework has been extensively refined \citep{pividori2023projecting, taroni2019multiplier, zhang2024mousiplier}, but  key problems remain. First, it restricts factor loadings to be positive, which rules out commonly observed and biologically meaningful negative correlations
\citep{zeng2010maximization, tu2015using}. Hence, PLIER estimates need ad hoc post-processing, which results in sub-optimal accuracy, as we will show in our numerical experiments. PLIER is highly sensitive to hyperparameter choices, requiring computationally expensive tuning and potentially leading to poor recovery of the empirical gene correlation structure, as illustrated in Figure~\ref{fig:method} for our motivating global fever dataset.
Furthermore, PLIER and related methods (e.g. \cite{buettner2017f}) do not provide  uncertainty quantification (UQ) in the inferred covariance matrix.

\begin{figure}[t]
\begin{center}
         \includegraphics[width=0.99\textwidth]{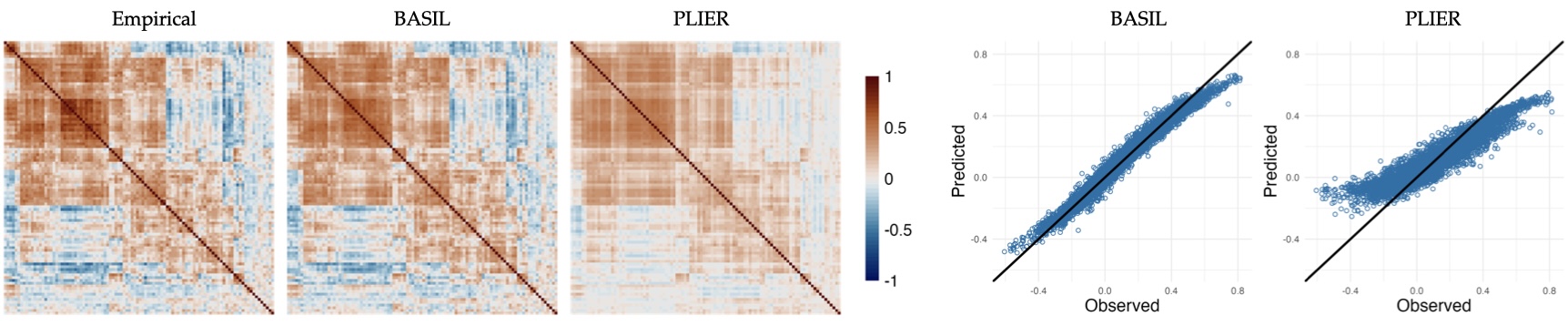} 
    \caption{Empirical gene expression correlation matrix computed from a subset of 100 genes 
    from a global fever study and correlation matrices estimated with BASIL and PLIER, reconstructed through latent variable factorization. The scatter plots show observed versus predicted correlation values.}
    \label{fig:method}
\end{center}
\end{figure}

To overcome these  limitations, we develop  Bayesian Analysis with gene-Sets Informed Latent space (BASIL). 
 BASIL is a factor model that incorporates the auxiliary information contained in the gene sets and adopts a Bayesian approach to inference. 
Bayesian methods offer a natural way to incorporate known gene sets through structured prior distributions and provide straightforward uncertainty quantification. The proposed approach projects high-dimensional gene expression data onto a low-dimensional set of latent factors, each representing an interpretable combination of gene sets. By integrating knowledge of gene sets, these latent variables capture structured variation in expression and provide meaningful insights into molecular processes. The number of latent factors is automatically selected by minimizing an information criterion, which consistently identified the correct number in our simulation studies. 

BASIL relies on a scalable and computationally efficient algorithm that enables fast inference even when analyzing 10,000s of genes. Posterior computation exploits a blessing of dimensionality to motivate pre-estimation of the latent factors \citep{fable, flair}, avoiding computationally expensive Markov chain Monte Carlo (MCMC) sampling. Related Bayesian factor modeling approaches that incorporate auxiliary or meta-covariate information have been proposed in environmental health studies \citep{bersson2025covariance}, ecology \citep{schiavon2022generalized}, and single-cell RNA sequencing data \citep{canale2023structured}. The \cite{canale2023structured} approach imposes structured shrinkage towards a diagonal covariance, which can lead to underestimation of important covariation patterns in the data.
Moreover, these methods rely on MCMC algorithms that become computationally prohibitive in high-dimensional 'omics settings. 

We show promising performance of BASIL through extensive simulation studies and analyzing two different transcriptomic datasets: the recent global fever dataset \citep{ko2023host}, which has not previously been analyzed using pathway-informed techniques, and the whole-blood RNA-seq data provided in \cite{mao2019pathway} and previously analyzed with PLIER.
The latter, while already examined in \cite{mao2019pathway}, is of particular interest because it constitutes the benchmark dataset for PLIER and represents a markedly high-dimensional setting, with only 36 samples and approximately 5,892 genes.
Such a pronounced $p \gg n$ setting is precisely where principled uncertainty quantification becomes most consequential as we will show in Section \ref{sec:WBanalysis}.
Overall, the results demonstrate that BASIL not only identifies biologically significant gene modules but also achieves superior predictive accuracy compared to state of the art approaches, while additionally providing principled uncertainty quantification that enables assessment of variability in the estimated latent factors.

\section{Bayesian Analysis with gene-Sets Informed Latent space}\label{sec:basil}

\subsection{Setup}\label{subsec:model}

Let $Y \in \mathbb{R}^{n\times p}$ be the gene expression profile, where $n$ is the number of samples, $p$ is the number of genes and usually $n \ll p$. Let $y_i = (y_{i1},\ldots,y_{ip})^{\top}$ denote the $i$-th row of $Y$ for $i=1,\dots,n$. 
We model the $i-$th sample as follows 
\begin{eqnarray}
y_i = \Lambda \eta_i + \epsilon_i,\quad 
\epsilon_i \sim N_p(0,\Sigma), \label{eq:factor1}
\end{eqnarray}
with $\eta_i = (\eta_{i1},\ldots,\eta_{ik})^{\top} \sim N_k(0, I_k)$  latent factors for patient $i$ where $k$ is the number of factors, 
$\Lambda$ is a $p \times k$ loadings  matrix relating these latent factors to the different gene expressions, and 
$\epsilon_i$ is a residual term having covariance matrix $\Sigma = \sigma^2 I_p$. Integrating out the latent factors from \eqref{eq:factor1}, we obtain an equivalent representation of the model as 
$$y_i \sim N_p(0, \Lambda \Lambda^\top + \sigma^2 I_p).$$
In particular, \eqref{eq:factor1} only requires $p(k+1) - {k \choose 2}$ effective parameters to model the $p \times p$ covariance of $y_i$ dramatically reducing the dimensionality of the problem.

\subsection{Biologically informed prior distributions}\label{subsec:prior}

In order to incorporate prior knowledge of gene sets into learning of gene expression factors, BASIL defines a hierarchical structure in the factor loading matrix $\Lambda$ in \eqref{eq:factor1}.
We denote by $C$ the matrix of $p \times q$ gene sets, where $q$ is the number of gene sets. 
Let $c_j=(c_{j1},\ldots,c_{jq})^{\top} \in \{0,1\}^q$, with $c_{jl}=1$ indicating that the gene $j$ belongs to the gene set $l$ for $j=1,\ldots,p$ and $l=1,\ldots,q$. 
 Specifically, we decompose the loading matrix as
\begin{eqnarray}
\Lambda = C \Gamma + \Psi, \label{eq:loads} 
\end{eqnarray}
where $C$ is the gene sets matrix, $\Gamma = \{ \gamma_{lh} \}$ is a $q \times k$ matrix linking gene expression factors to gene sets ($\gamma_{lh} \neq 0$ if factor $h$ is linked to gene set $l$), and $\Psi$ is an unstructured $p \times k$ matrix characterizing the variation in the gene expressions not related to pre-existing gene sets. 
To enforce identifiability, up to orthogonal transformation, we require $\Psi \in \mathcal N(C)$, where $\mathcal N (C)$ denotes the null space of $C$.
Expressions \eqref{eq:factor1}-\eqref{eq:loads} imply the following covariance structure in the gene expression data
\begin{eqnarray*}
\mbox{cov}(y_i) = [C\Gamma + \Psi][\Gamma^{\top}C^{\top} + \Psi^{\top}] + \Sigma = C\Gamma \Gamma^\top  C^\top + \Psi \Psi^\top + \Sigma , \label{eq:cov}    
\end{eqnarray*}
which is the sum of two low-rank components and a diagonal matrix. 
Hence, the covariance matrix depends on the matrix of gene sets $C$, which guides the inference.
In many cases, including our motivating application, we expect pathways to carry abundant information, with $C$ explaining most of the variability, and, hence, $\Lambda \approx C \Gamma$. Nevertheless, knowledge of pathways is often incomplete, and $\Psi$ models the effect of unknown components. This structure allows for the incorporation of known gene ontology information while addressing the fact that many genes are not in any existing gene set. This represents a key advantage of BASIL over PLIER \citep{mao2019pathway}, which does not explicitly distinguish between structured and unstructured sources of variability.

Given the structure in \eqref{eq:loads}, one could naively endow entries $\Gamma$ and $\Psi$ with conditionally Gaussian priors and perform posterior computation via Gibbs sampling. We identify two main difficulties related to such an approach: (i) the orthogonality condition would not be respected, hindering the interpretability of the results and potentially worsening mixing of the MCMC sampling, and (ii) the complete conditional of $\Gamma$ would depend on the entire data matrix $Y$ due to the lack of conditional independence for the presence of $C$, substantially increasing its computational cost.
These considerations motivated us to consider an alternative specification. In particular, let $B_{\mathcal C} \in \mathbb R^{p \times q}$ and  $B_{\mathcal N} \in \mathbb R^{p \times p-q}$ denote orthonormal bases of the column space and null space of $C$, which we denote by $\mathcal C(C)$ and $\mathcal N(C)$, respectively. A canonical choice of $B_{\mathcal C}$ is the matrix of left singular vectors of $C$ associated to its $q$ non-zero singular values. The formulation in 
\eqref{eq:loads} can be equivalently represented as 
\begin{eqnarray}
\Lambda = B_{\mathcal C} \Lambda_{\mathcal C} + B_{\mathcal N} \Lambda_{\mathcal N}, \label{eq:loads2} \end{eqnarray}
where $\Lambda_{\mathcal C}  = B_{\mathcal C}^\top \Lambda = B_{\mathcal C}^\top C \Gamma \in \mathbb R^{q \times k}$ and  $\Lambda_{\mathcal N} = B_{\mathcal N}^\top \Lambda = B_{\mathcal N}^\top \Psi \in \mathbb R^{p-q \times k}$ are the linear combinations of the components $\Lambda$ with respect to the choice of basis for $\mathcal C(C)$ and $\mathcal N (C)$ respectively.  
This automatically enforces orthogonality between the two components and, as we show in the next section, allows a simple and computationally efficient posterior update.

We choose shrinkage priors for the elements of $\Lambda_{\mathcal C}$ and $\Lambda_{\mathcal N}$:
\begin{equation*}\label{eq:prior_lambda}
    \begin{aligned}
        \lambda_{\mathcal C l} \mid \sigma^2 &\sim \mathcal N_k(0, \tau_{\Gamma}^2 \sigma^2 I_k), \quad (l=1, \dots, q),\\ 
        \lambda_{\mathcal N l'}\mid \sigma^2 &\sim \mathcal N_k(0, \tau_{\Psi}^2 \sigma^2 I_k), \quad (l'=1, \dots, p-q),\\
    \end{aligned}
\end{equation*}
with $  \lambda_{\mathcal C l}$ and $\lambda_{\mathcal N l'}$ denoting the $l$-th and $l'$-th rows of $\Lambda_{\mathcal C}$ and $\Lambda_{\mathcal N}$, respectively.

Importantly, we allow for different degrees of shrinkage in $\Lambda_{\mathcal C}$ and $\Lambda_{\mathcal N}$ to accommodate cases where the strength of the signal explained by $C$ differs from the residual one, via the terms $\tau_{\Gamma}^2$ and $\tau_{\Psi}^2$. This structure avoids strong assumptions about the underlying biological processes while enabling BASIL to adaptively infer gene–pathway relationships from the data.
To enhance data-adaptivity and limit the practitioner's tuning effort, we develop an empirical Bayes \citep{morris1983parametric} strategy to select the prior shrinkage coefficients $\tau_{\Gamma}$ and $\tau_{\Psi}$.  For instance, we show that in absence of residual signal, BASIL collapses on the simpler model $\Lambda = C \Gamma$ as $\tau_{\Psi} \to 0$. This allows the data to inform about the degree to which the inferred lower-dimensional factors $\eta_i$ underlying the high-dimensional $y_i$ are linked to pre-existing gene sets. The structure also enables learning of new gene sets corresponding to genes that have non-zero loadings $\Lambda$. 

 We complete the prior specification with an inverse-gamma prior on $\sigma^2$, i.e. $ \sigma^2 \sim IG(v_0/2, \sigma_0^2v_0/2).$

\subsection{Posterior computation}

The uncertainty in $\Lambda_{\mathcal C}$ and $\Lambda_{\mathcal N}$ is encoded in their posterior distribution
\begin{equation*}    \label{eq:true_posterior}
\begin{aligned}
    p(\Lambda_{\mathcal C}, \Lambda_{\mathcal N} \mid Y, X) 
    & \propto \int p(Y \mid \Lambda_{\mathcal C}, \Lambda_{\mathcal N},  M) p(\Lambda_{\mathcal C}) p(\Lambda_{\mathcal N}) p(M) d M \\
    & \propto \int p(\Lambda_{\mathcal C}, \Lambda_{\mathcal N} \mid Y, M)p(M \mid Y)dM. 
\end{aligned}
\end{equation*}
Bayesian methods generally rely on MCMC sampling or variational inference for posterior computation. MCMC algorithms are computationally expensive, while variational approximations often massively underestimate uncertainty, hindering the reliability of interval estimates. Here, 
to facilitate efficient and accurate posterior approximation, we exploit a pre-training approach. We start with an initial estimate of $M=[\eta_1 ~ \cdots  ~  \eta_n]^\top$ given by $\hat M = \sqrt{n} U$, where $U \in \mathbb R^{n \times k}$ is the matrix of left singular vectors of $Y$ associated with its $k$ leading singular values. This is sometimes referred to as a PCA-estimate and is commonly used in the econometrics literature \citep{bai_03}. Pre-training for efficient posterior approximation in factor models was initially proposed in 
 \citet{fable}, and further developed in \citet{flair, fama, blast}, but none of these approaches can include prior information on gene sets. Intuitively, as the number of variables $p$ increases, the marginal posterior of latent factors concentrates around a point estimate, so ignoring uncertainty in $M$ has negligible impact on accuracy of uncertainty quantification and estimation. 
 \citet{fable} provide a rigorous treatment of such an approach including posterior contraction rates around the true parameter values. 
 This approach is related to joint estimation for factor models \citep{chen_jmle, chen_identfiability, lee24}, which is consistent in double asymptotic scenarios in which both sample size and outcome dimension diverge. 
 Hence, we approximate the posterior distribution of $\Lambda_{\mathcal C}$ and $\Lambda_{\mathcal N}$ via their conditional posterior distribution given the estimate $\hat M$ for $M$ as follows:
\begin{equation*}\label{eq:posterior_approx_1}
     p(\Lambda_{\mathcal C}, \Lambda_{\mathcal N} \mid Y) \approx p(\Lambda_{\mathcal C}, \Lambda_{\mathcal N} \mid Y, \hat M).
\end{equation*}
 
Conditionally on $\hat M$, inference on $\Lambda_{\mathcal C}$ and $\Lambda_{\mathcal N}$ corresponds to surrogate multilinear regression tasks, where $\hat M$ is treated as the observed covariates matrix. Overall, the inferential procedure is efficient and extremely fast,  while producing intervals with correctly calibrated coverage.

Let $Y_{\mathcal C} = Y B_{\mathcal C}$ and $Y_{\mathcal N} = Y B_{\mathcal N}$. It is easy to see that for their $i$-th rows, $y_{\mathcal C i} = B_{\mathcal C}^\top y_i$ and $y_{\mathcal Ni} = B_{\mathcal N}^\top y_i$, the following holds: 
\begin{equation*}
    y_{\mathcal C i} = \Lambda_{\mathcal C} \eta_i + \epsilon_{\mathcal C i}, \quad  y_{\mathcal N i} = \Lambda_{\mathcal N} \eta_i + \epsilon_{\mathcal N i},
\end{equation*}
where $\epsilon_{\mathcal C i} \sim N_q(0, \sigma^2 I_q)$, $\epsilon_{\mathcal N i} \sim N_{p-q}(0, \sigma^2 I_{p-q})$, $y_{\mathcal C i} \independent y_{\mathcal N i}$, and denote the corresponding data matrices with $Y_{\mathcal C}$.
Our prior specification for the elements of $\Lambda$ and $\sigma^2$ allows a conjugate update:
\begin{equation*}\label{eq:posterior_update}\begin{aligned}
\sigma^2 &\mid Y_{\mathcal C}, Y_{\mathcal N}  \sim IG\left(\frac{v_n}{2}, \frac{v_n \sigma_n^2}{2}\right),\\
 \lambda_{\mathcal C l} &\mid  Y_{\mathcal C}, \sigma^2 \sim N_k\left(\mu_{\mathcal Cl},  \frac{\sigma^2 }{n + \tau_{\Gamma}^{-2}} I_k\right), \quad (l = 1, \dots, q),\\
  \lambda_{\mathcal N l'} &\mid  Y_{\mathcal N}, \sigma^2  \sim N_k\left(\mu_{\mathcal N l'},  \frac{\sigma^2 }{n + \tau_{\Psi}^{-2}} I_k\right), \quad (l' = 1, \dots, p-q),\\
\end{aligned}
\end{equation*}
where 
\begin{equation}\label{eq:posterior_params}
    \begin{aligned}
    v_{n} &= v_0 + n,\\
    \sigma_n^2 &= \frac{1}{v_n} \left\{v_0 \sigma_0^2  + \sum_{l=1}^q \big[y_{\mathcal C}^{(l)\top } y_{\mathcal C}^{(l)} -  (n + \tau_{\Gamma}^{-2})\mu_{\mathcal C l}^\top \mu_{\mathcal C l}\big] + \sum_{l'=1}^{p-q} \big[y_{\mathcal N}^{(l')\top } y_{\mathcal N}^{(l')} -  (n + \tau_{\Psi}^{-2})\mu_{\mathcal N l'}^\top \mu_{\mathcal N l'}\big] \right\},\\
         \mu_{\mathcal C l} &=  \left(\hat M^\top \hat M + \tau_{\Gamma}^{-2} I_{k} \right)^{-1}\hat M^{\top} y_{\mathcal C}^{(l)} = \frac{1}{n + \tau_{\Gamma}^{-2}} \hat M^\top y_{\mathcal C}^{(l)}, \\
            \mu_{\mathcal N l'} &=  \left(\hat M^\top \hat M + \tau_{\Psi}^{-2} I_{k} \right)^{-1}\hat M^{\top} y_{\mathcal N}^{(l')} = \frac{1}{n + \tau_{\Psi}^{-2}} \hat M^\top y_{\mathcal N}^{(l')}, \\    
    \end{aligned}
\end{equation}
with $y_{\mathcal C}^{(l)}$ and $y_{\mathcal N}^{(l')}$ denoting the $l$-th and $l'$-th columns of $Y_{\mathcal C}$ and $Y_{\mathcal N}$.  Importantly, this step can be parallelized across columns of $Y_{\mathcal C}$ and $Y_{\mathcal N}$. Moreover, because of conjugacy, we can obtain independent samples from the posterior of the loadings without the need for MCMC. As in other pre-training methods, this approach induces a posterior distribution over components of the covariance matrix that suffers from a mild undercoverage.  We solve this issue by applying a simple coverage correction strategy, as in \citet{fable}, consisting of inflating the posterior variance of $\Lambda_{\mathcal C}$ and $\Lambda_{\mathcal N}$ by the term $\rho^2>1$, which is tuned to achieve valid asymptotic frequentist coverage. The coverage corrected samples for $\Lambda$ are then given by 
\begin{equation}\label{eq:posterior_update_cc}\begin{aligned}
 \lambda_{\mathcal C l} &\mid  Y_{\mathcal C}, \sigma^2 \sim N_k\left(\mu_{\mathcal Cl},  \frac{\sigma^2 \rho^2}{n + \tau_{\Gamma}^{-2}} I_k\right), \quad (l = 1, \dots, q)\\
  \lambda_{\mathcal N l'} &\mid  Y_{\mathcal N}, \sigma^2  \sim N_k\left(\mu_{\mathcal N l'},  \frac{\sigma^2 \rho^2}{n + \tau_{\Psi}^{-2}} I_k\right), \quad (l' = 1, \dots, p-q).\\
\end{aligned}
\end{equation}
The posterior mean for $\Lambda_{\mathcal C}$ and $\Lambda_{\mathcal N}$ are available in closed form: 
\begin{equation}\label{eq:posterior_mean_lambda}\begin{aligned}
 \mathbb E[ \Lambda_{\mathcal C} \mid -] &= \begin{bmatrix}\mu_{\mathcal C 1}  ~ \cdots  ~\mu_{\mathcal C q} 
 \end{bmatrix}^\top = \frac{\sqrt{n}}{n +  \tau_{\Gamma}^{-2}} Y_{\mathcal C}^\top U = \frac{\sqrt{n}}{n +  \tau_{\Gamma}^{-2}} B_{\mathcal C}^\top V D,\\
  \mathbb E[ \Lambda_{\mathcal N} \mid -] &= \begin{bmatrix}\mu_{\mathcal N 1}  ~ \cdots  ~\mu_{\mathcal N p-q} 
 \end{bmatrix}^\top = \frac{\sqrt{n}}{n +  \tau_{\Psi}^{-2}} Y_{\mathcal N}^\top U = \frac{\sqrt{n}}{n +  \tau_{\Psi}^{-2}} B_{\mathcal N}^\top V D,
\end{aligned}
\end{equation}
and the induced posterior means for $\Gamma$ and $\Psi$ are 
\begin{equation}\label{eq:posterior_mean_lambda_2}
\begin{aligned}
     \mathbb E[\Gamma \mid -] &= \frac{\sqrt{n}}{n + \hat \tau_{\Gamma}^{-2}} (C^\top C)^{-1}C^{\top} V D,\\
      \mathbb E[\Psi \mid -] &= \frac{\sqrt{n}}{n + \hat \tau_{\Psi}^{-2}}(I - P_C) V D,
\end{aligned}   
    \end{equation} 
    where  
$D = \text{diag}(d_1, \dots, d_k)$, with $d_l$ being the $l$-th largest singular value of $Y$, and $V \in \mathbb R^{p\times k}$ is the matrix of the associated right singular vectors of $Y$. It is easy to show that the induced posterior mean on the covariance matrix is given by
\begin{equation}\label{eq:posterior_mean_Lambda_outer}
    \mathbb E[\Lambda \Lambda^\top + \sigma^2 I_p \mid -] = \bar \Lambda \bar \Lambda^\top + \sigma_n^2 \frac{v_n}{v_n -2} \bigg\{1 + \rho^2\Big(\frac{1}{n + \tau_\Gamma^{-2}} + \frac{1}{n + \tau_\Psi^{-2}} \Big)\bigg\} I_p,
\end{equation}
where $\bar \Lambda = \mathbb E[\Lambda \mid -] = C\mathbb E[\Gamma \mid -] + \mathbb E[\Psi \mid -]$.
Given a sample for $\Lambda_{\mathcal C}$ and $\Lambda_{\mathcal N}$, we obtain a sample for $\Lambda $ via \eqref{eq:loads2}. These samples are drawn independently, substantially improving the computational efficiency over MCMC. Posterior summaries for any functional of the model parameters, including point and interval estimates, can be calculated using such samples. For instance, for given samples of $\Lambda$ and $\sigma$, we can sample the latent factors of the $i$-th sample (the $i$-th row of $M$), via 
\begin{equation*}
    \eta_i \mid y_i, \Lambda, \sigma^2 \sim N_k\bigg( \big(\Lambda^\top \Lambda + \sigma^2I_k\big)^{-1}\Lambda^\top y_i,   \big(\frac{1}{\sigma^2}\Lambda^\top \Lambda + I_k\big)^{-1} \bigg).
\end{equation*}
Given $N_{MC}$ samples of $\Lambda$ and $\sigma^2$, $\{\Lambda^{(s)}\}_{s=1}^{N_{MC}}, \{\sigma^{2(s)}\}_{s=1}^{N_{MC}}$, we can approximate the posterior mean for $\eta_i$ via
\begin{equation*}
    \mathbb E[\eta_i \mid y_i] \approx \frac{1}{N_{MC}}\sum_{s=1}^{N_{MC}}   \big(\Lambda^{(s)\top} \Lambda^{(s)} + \sigma^{2(s)} I_k\big)^{-1}\Lambda^{(s)\top} y_i.
\end{equation*}

The complete procedure to obtain point and intervals estimates is reported in Algorithm \ref{alg:basil} in the supplementary material.

\subsection{Hyperparameter tuning}\label{subsec:hyperparameters}

This section describes how the hyperparameters of BASIL are selected. In this regard BASIL offers a major advantage over PLIER \citep{mao2019pathway} by avoiding the need for intensive hyperparameter tuning and making the inference process more efficient, stable, and user-friendly. 
As a criterion to select the number of factors, we use the joint likelihood based information criterion (JIC) of \citet{chen_jic},
\begin{equation}\label{eq:JIC}
    \text{JIC}(k) = -2 l_{k} + k \max(n,p) \log \{\min(n,p)\},
\end{equation}
where $l_{k}$ is the value of the joint log-likelihood computed at the joint maximum likelihood estimate for latent factors and factor loadings when the latent dimension is equal to $k$. To avoid the computation of joint maximum likelihood estimates for each value of $k$, we approximate $l_{k}$ with $ l_{k} \approx \hat l_{k} $, where $\hat l_{k}$ is the likelihood for the study-specific data matrix under spectral estimates, i.e. the product of the estimates for latent factors and loadings correspond to the best rank $k$ approximation to $Y$. 

For the prior shrinkage coefficients, note that conditionally on $\sigma^2$, the \textit{a priori} expected values of the Frobenius norm squared of $ \Lambda_{\mathcal C}$ and $ \Lambda_{\mathcal N}$ are $\mathbb E[||\Lambda_{\mathcal C}||_F^2 \mid \sigma^2] = \tau_{\Gamma}^2 \sigma^2 qk$ and $\mathbb E[||\Lambda_{\mathcal N}||_F^2 \mid \sigma^2] = \tau_{\Psi}^2 \sigma^2 (p-q)k$.
We let $L_{\mathcal C}= || P_C V D||_F^2 /n$, $L_{\mathcal N}= || (I_p - P_C) V D||_F^2 /n$, and $\hat \sigma^2 = ||(I_n - U U^{\top}) Y ||_2^2/[(n - k) p]$, which are consistent estimates of $||\Lambda_{\mathcal C}||_F^2$, $||\Lambda_{\mathcal N}||_F^2$, and $\sigma^2$, respectively. Thus, we set $\tau_{\Gamma}^2$ and $\tau_{\Psi}^2$ to $\hat \tau_{\Gamma}^2 = \frac{L_{\mathcal C}}{k q \hat \sigma^2}$ and  $\hat \tau_{\Psi}^2 = \frac{L_{\mathcal N}}{k (p-q) \hat \sigma^2}$, respectively. We discuss the theoretical properties for such hyperparameter values in the next section.
For $\rho$, we take inspiration from \citet{fable}, and for $f,g = 1, \dots, p$, we let \begin{equation*} b_{ fg} = \begin{cases}    \big\{ 1 +  \frac{||\mu_{ f}||_2^2 ||\mu_{ g} ||_2^2 + (\mu_{\mathcal C f}^\top\mu_{ g})^2 }{ \hat \sigma^2 (||\mu_{ f}||_2^2 + ||\mu_{ g}||_2^2)} \big\}^{1/2}, \quad &\text{if } f \neq g  \\   \big(1 + \frac{ ||\mu_{ f} ||_2^2 }{ 2\hat \sigma^2 }\big)^{1/2}, \quad &\text{otherwise,}         \end{cases}
\end{equation*}
where $\mu_j$ is the $j$-th row of $\mathbb{E}[\Lambda \mid -] = C  \mathbb E[\Gamma \mid -] +  \mathbb E[\Psi \mid -]$.
Then, we set $\rho= {q \choose 2}^{-1} \sum_{1 \leq f \leq g \leq p}  b_{fg}$, which allows control of the coverage on average across entries. Finally, we set $v_0=1$ and $\sigma_0^2 = 1$.

\subsection{Theoretical support}\label{eq:theory}

In this section, we illustrate some benefits of the double shrinkage approach described in Section \ref{subsec:prior}. In particular, we show that in the limiting cases when only one source of variation (either the known or the novel pathways) is active, the variance parameter of the other component converges to $0$. This allows the model to collapse on a simpler specification where only the non-superfluous component is included. We accompany these theoretical results with a simulation study demonstrating the adaptivity of the shrinkage parameters in intermediate regimes where both sources of variation contribute (see Figure \ref{fig:low_high}).

We start by introducing the notation. For a matrix $A$, we denote by 
$s_l(A)$ its $l$-th largest singular value. 
Moreover, for two sequences $(a_n)_{n \leq 1}$, $(b_n)_{n \leq 1}$, we say $a_n \lesssim b_n$ if there exist two constants $N_0 < \infty$ and $C < \infty$, such that $a_n \leq C b_n$ for every $n > N_0$. We say $a_n \asymp b_n$ if and only if $a_n \lesssim b_n$ and $b_n \lesssim a_n$.

We enumerate a few regularity conditions. 
\begin{assumption}\label{assumption:data}
    Data are generated under model \eqref{eq:factor1}, with true loading matrix $\Lambda_0 = C \Gamma_0 + \Psi_0$, 
    and error variance $\sigma_0^2$. Moreover, we have $C \Gamma_0 \perp \Psi_0$.     
\end{assumption}
The condition $C \Gamma_0 \perp \Psi_0$ in Assumption \ref{assumption:data} is needed to ensure the identifiability of the parameters $\Gamma_0$ and $\Psi_0$. Indeed if such a condition was not imposed, for some given values of the parameters, $\Gamma_0$ and $\Psi_0$, one could let $\tilde \Gamma_0 = 0$ and $\tilde \Psi_0 = \Psi_0 + C \Gamma_0$. It is easy to verify that $(\Gamma_0, \Psi_0)$ and $(\tilde \Gamma_0, \tilde \Psi_0)$ induce the same distribution for $y$. Imposing this condition resolves this non-identifiability due to the uniqueness of the decomposition of vectors as a sum of elements lying in orthogonal subspaces.

\begin{assumption}\label{eq:dimension}
    We let $p = p_n \to \infty$ and $\log(p_n )/n = \mathcal O(1)$ as $n \to \infty$. \end{assumption} 
Assumption \ref{eq:dimension} requires the dimensionality to grow with the sample size no faster than a polynomial rate. Considering a growing dimension is reasonable since in genomics we typically have $p \gg n$. 

\begin{assumption}\label{assumption:Lambda}
    We have $s_l(\Lambda_0)\asymp \sqrt{p} $, which implies that either $s_l(C \Gamma_0)  \asymp \sqrt{p} $ or $s_l(\Psi_0) \asymp \sqrt{p}$, for $l=1, \dots, k$.
\end{assumption}
Assumption \ref{assumption:Lambda} allows the systematic component of the variation to be distinguished from pure noise.
\begin{assumption}\label{assumption:hyperparameters}
    The number of latent factors $k$ is known.
\end{assumption}
While this Assumption may not hold in practice, the JIC defined in \eqref{eq:JIC} has been shown to provide consistent estimates for $k$ under mild assumptions \citep{chen_jic}. We note that in the numerical experiments we do not use any knowledge about the true number of latent factors and obtain good estimation accuracy. 

We characterize the behaviour of the prior shrinkage parameters in two special cases, that is when all the variation is explained by either known or new pathways.
\begin{theorem} \label{thm:prior_variances}
    Under Assumptions \ref{assumption:data} to \ref{assumption:hyperparameters}, as $n \to \infty$, if $\Psi_0 = 0$, or, equivalently, there are no new pathways active, $\hat \tau_{\Psi} \to 0$ in probability, and, if $\Gamma_0 = 0$, or, equivalently, there are no known pathways active, $\hat \tau_{\Gamma} \to 0$ in probability.
\end{theorem}
Theorem \ref{thm:prior_variances} characterizes the behavior of BASIL in two specific scenarios. When there are no new pathways active and all the systematic variability is explained by known pathways, or, formally, $\Psi_0 = 0$, the prior variance parameter controlling the shrinkage of $\Psi$ converges to $0$, and BASIL collapses on the simpler model $\Lambda = C \Gamma$. Analogously, when all variability is explained by new pathways, $\Gamma_0 = 0$, the prior variance parameter controlling the shrinkage of $\Gamma$ converges to $0$, and BASIL collapses on the simpler model $\Lambda = \Psi$. Therefore, our procedure adaptively selects the amount of shrinkage to apply to structured and unstructured variability and chooses simpler model specifications when adequate.

\begin{figure}[t]
\begin{center}
    \includegraphics[width=0.6\textwidth]{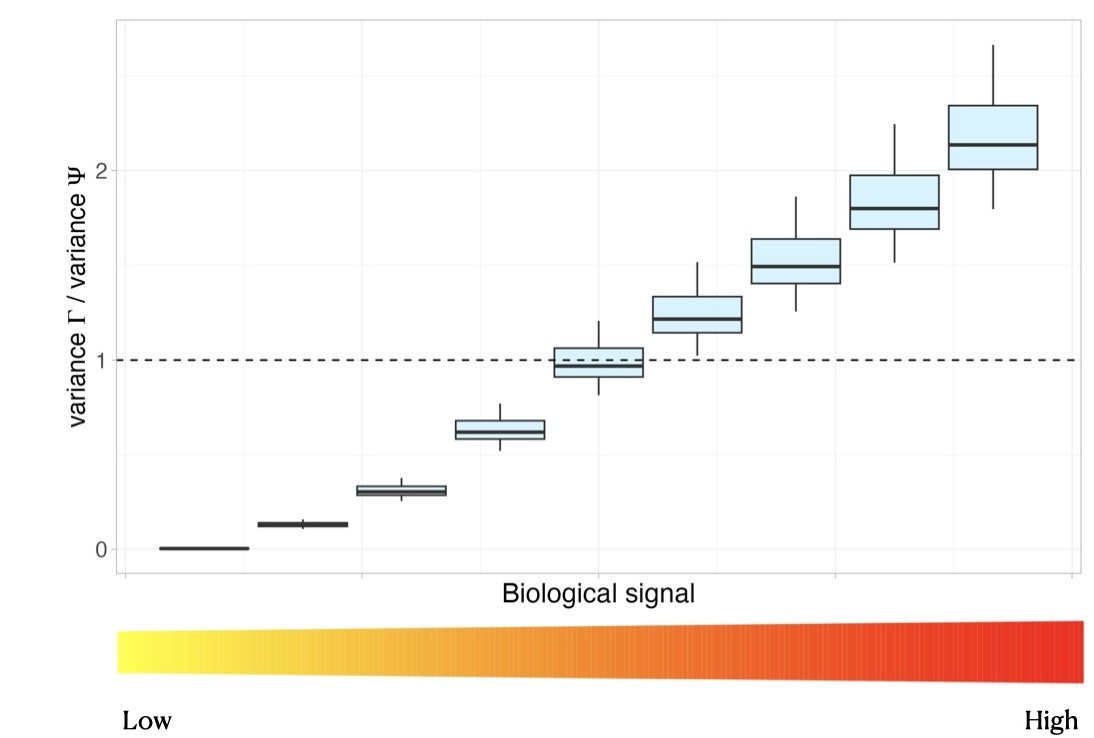} 
    \caption{Ratio between the estimated variances for $\Gamma$ and $\Psi$ as the ontology signal increases. The plot shows the results over $25$ replications for each setting. }
    \label{fig:low_high}
\end{center}
\end{figure}

The adaptive shrinkage behavior of BASIL is also evident in simulations that explore scenarios between the two extreme cases described above. We simulate synthetic data under the BASIL model \eqref{eq:factor1}–\eqref{eq:loads}, fixing $\tau_{\Psi}^2$ while varying $\tau_{\Gamma}^2$ to generate datasets with different levels of biological signal.
The ratio between $\tau_{\Gamma}^2$ and $\tau_{\Psi}^2$ serves as an indicator of how informative gene ontology knowledge is in the latent factor model. A ratio of less than one indicates substantial unstructured variability, while a ratio above one implies that the majority of the signal in the data aligns to known gene ontologies. 
As shown in Figure~\ref{fig:low_high}, BASIL’s estimates of this ratio increase as the ontology-driven signal becomes stronger, illustrating its ability to adapt to the true underlying structure. These results confirm that BASIL automatically modulates the degree of shrinkage applied to structured and unstructured components, consistent with the theoretical behavior described in Theorem~\ref{thm:prior_variances}.

\section{Simulation study}\label{sec:simulations}

 We evaluated the performance of BASIL compared to state of the art methods across several metrics: estimation accuracy, computational efficiency, uncertainty quantification, and parameter interpretability. 
We simulated gene expression data in two different scenarios. The first, referred to as `high ontology signal', reflects a setting in which the variability in the loadings matrix $\Lambda$ is primarily driven by gene ontology information, with low unstructured variability. The second scenario,  referred to as `low ontology signal',  is dominated by unstructured variability, accounting for the fact
that biological signal in many genes is not fully characterized by known gene sets. The former represents an ideal case where known gene sets strongly inform latent gene associations, while the latter captures realistic settings where gene ontology information may be incomplete or inaccurate. 
In practice, this is implemented by assigning different values to the variance parameters associated with the structured ($\Gamma$) and unstructured ($\Psi$) components of the gene loadings matrix.
Specifically, we simulate data under the BASIL model \eqref{eq:factor1}-\eqref{eq:loads} and set $\tau_{\Gamma}^2=0.7$, $\tau_{\Psi}^2=0.1$ in the `high ontology signal' scenario and  $\tau_{\Gamma}^2=0.4$, $\tau_{\Psi}^2=0.7$ in the `low ontology signal' scenario. 
In both settings, we considered two different values for the number of genes - $1,000$ and $3,000$ -  $10$ latent factors, $500$ samples and set $\sigma^2=15$.
For each configuration, we reproduce the experiments 25 times. We obtain the matrix of gene sets $C$ using the \texttt{msigdbr} R package that provides access to gene sets from the Molecular Signatures Database, employing the molecular function subcategory of the gene ontology category. 

To assess the performance of BASIL in estimation accuracy, we evaluate the error in reconstructing 
$\Lambda_0 \Lambda_0^{\top}$ via the Frobenius norm of the difference of the estimate and true parameter scaled by $||\Lambda_0 \Lambda_0^{\top}||$. 
For comparison, we also ran PLIER \citep{mao2019pathway} and ROTATE \citep{rovckova2016fast}, a computationally efficient factor model based on an expectation-maximization (EM) algorithm that cannot incorporate information on gene ontology.
For PLIER, we used the default settings for hyperparameter selection and estimated the low rank component of the gene covariance matrix with 
$(\hat \Lambda (\hat M^\top \hat M)\hat \Lambda^{\top})/n$, where $\hat \Lambda$ and $\hat M$ are estimates for loadings and latent factors, respectively. Since ROTATE does not provide a default mechanism for selecting the number of factors, we set it to the true number, thereby giving this method an advantage. The performance of ROTATE tended to degrade when the number of factors was set at a value greater than $k$ in our preliminary experiments.

BASIL outperforms competitors in all scenarios (Figure~\ref{fig:sim1}, panel a) in recovering the true low-rank components of the gene covariance matrix $\Lambda_0 \Lambda_0^{\top}$, showing attractive results in latent factor inference. In addition to being very accurate, BASIL is also fast to run, especially compared to PLIER, as shown in Figure~\ref{fig:sim1}, panel b. Indeed, BASIL runs in a few seconds on the simulated datasets and is approximately 30 times faster than PLIER on average for our experiments.

\begin{figure}[t]
\begin{center}
    \includegraphics[width=0.99\textwidth]{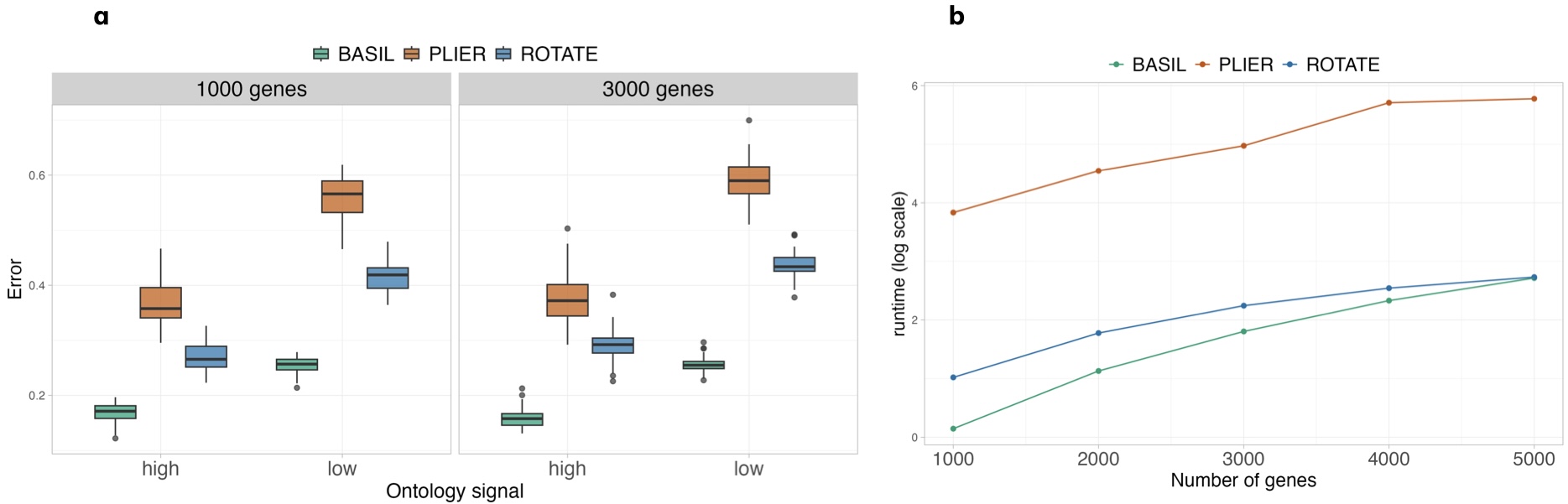} 
    \caption{ Panel a: performance of BASIL, ROTATE and PLIER in estimation accuracy for the gene co-expression matrix in low and high ontology signal scenarios over the $25$ replications.  Panel b: mean runtime over $25$ replications of BASIL, ROTATE and PLIER with datasets simulated under different values for the number of genes in the `high ontology signal' scenario. The results are in seconds on a logarithmic scale.}
    \label{fig:sim1}
\end{center}
\end{figure}

A critical issue in low-rank matrix factorization is the choice of the latent dimension, i.e. the number of factors. In our experiments, BASIL consistently recovered the true number of factors in all scenarios, while PLIER significantly overestimated it, as shown in Table \ref{tab:ksim}. This tendency toward overestimation not only introduces bias in factor estimation but also results in less parsimonious models, making the biological interpretation of the factors more challenging.

To assess BASIL's performance in uncertainty quantification, we evaluate the frequentist coverage of the posterior credible intervals for $\Lambda\Lambda^{\top}$. Specifically, we consider the frequentist coverage of 95\% credible intervals on average across the entries of a $200 \times 200$ submatrix of $\Lambda\Lambda^{\top}$.
We recall that both PLIER and ROTATE produce only point estimates of the gene covariance matrix, without offering any means to quantify uncertainty. For this reason, we exclude them from the uncertainty quantification experiments. The results are shown in Table \ref{tab:ksim}. 
Across all scenarios, the average coverage of the entrywise intervals is close to the nominal value 0.95, providing support for BASIL in terms of producing well-calibrated credible intervals.

\begin{table}[t]
\centering
\caption{BASIL and PLIER estimations of the number of latent factors $k$ across the two scenarios. For BASIL we also include frequentist coverage of 95\% credible intervals. We report mean and standard errors over 25 replications. 
}
\vspace{0.3cm}

\begin{tabular}{c|cc|cc}
\multicolumn{1}{c}{} &
\multicolumn{2}{c}{High bio signal} &
\multicolumn{2}{c}{Low bio signal} \\
& $p=1000$ & $p=3000$ & $p=1000$ & $p=3000$ \\
\hline
\multicolumn{5}{c}{\textit{Estimated $k$ }} \\
\hline
BASIL & 10 (0) & 10 (0) & 10 (0) & 10 (0) \\
PLIER &  69 (27) & 51 (16)& 61 (24) & 58 (24)\\
\hline
\multicolumn{5}{c}{\textit{Coverage — BASIL only}} \\
\hline
BASIL & 0.93 (0.02) & 0.93 (0.02) & 0.93 (0.01) & 0.93 (0.01) \\
\hline
\end{tabular}

\label{tab:ksim}
\end{table}

\section{Application to whole-blood RNA-seq data} \label{sec:WBanalysis}

We apply BASIL to the validation dataset provided in \cite{mao2019pathway}, which comprises human whole-blood samples assayed by both RNA sequencing and direct mass cytometry measurement (using CyTOF) of cell-type proportion.
The dataset consists of 36 samples for 5892 genes. The pathway matrix $C$ contains 606 pathways that included 61 cell-type markers and 555 canonical pathways from the Molecular Signatures Database (MSigDB). For details on sample processing and RNA sequencing methods refer to \cite{mao2019pathway}.  Data are standardized to have mean zero and variance 1 for each column (gene). The criterion in \eqref{eq:JIC} selects $k=8$ factors and the estimates of the shrinkage parameters $\tau_{\Gamma}^2$ and $\tau_{\Psi}^2$ are 0.82 and 0.29, respectively, suggesting that pathway information plays an important role and should be considered in the model.

 The high dimensionality of the dataset combined with the low sample size induces non-negligible uncertainty in the inferred gene-gene covariance matrix. Methods like PLIER neglect such uncertainty, while BASIL offers a way to seamlessly focus on relevant associations, and exclude spurious ones, via the posterior distribution of quantities of interest. For example, Figure \ref{fig:appWB}a shows the empirical correlation matrix for the subset of the 100 genes with the largest empirical variance and the estimates of BASIL and PLIER. In the third panel, we display the BASIL posterior mean of the correlation matrix, after zeroing out the entries for which the corresponding $95\%$ credible interval included zero. Roughly $85\%$ of the possible pairs of genes did not display a statistically significant correlation value. This dramatically decreases the number of associations to be considered, facilitating downstream analysis and interpretability. 
 
\begin{figure}[t!]
\begin{center}
    \includegraphics[width=0.9\textwidth]{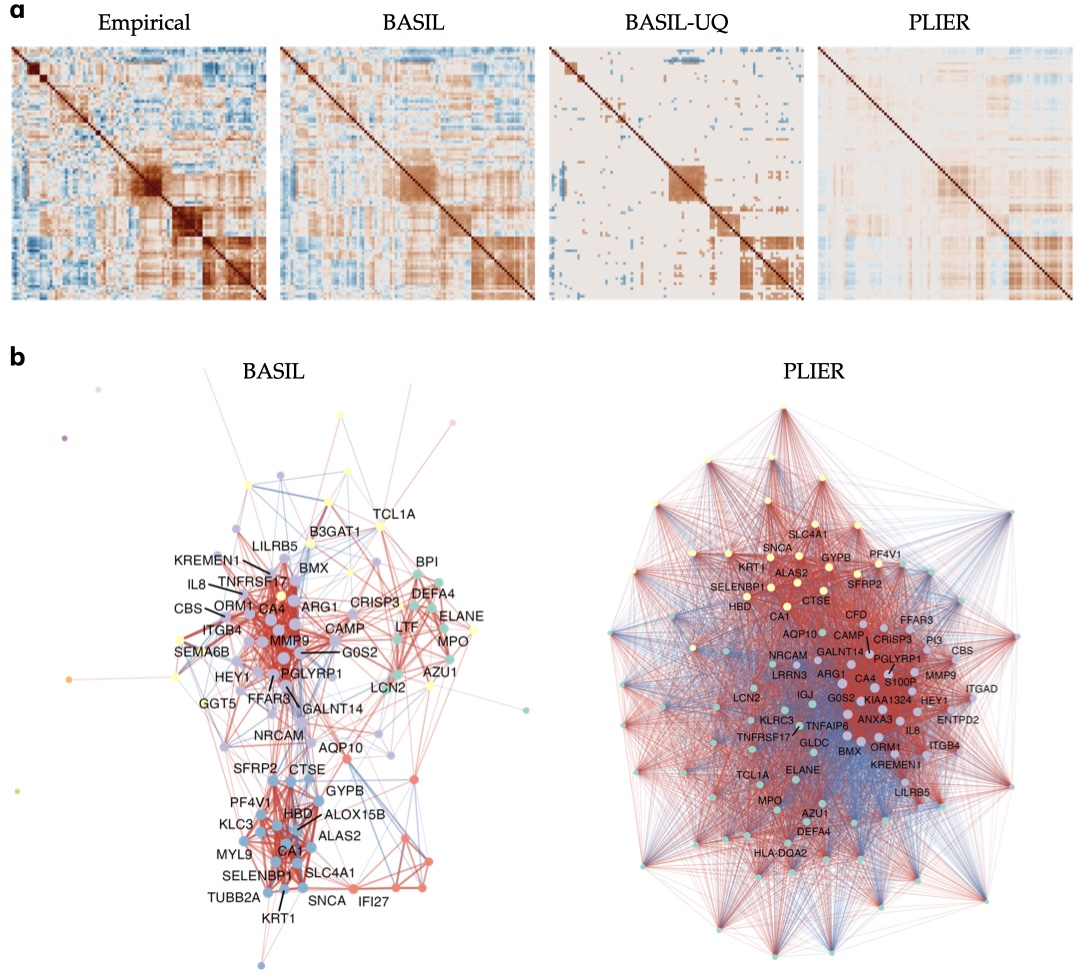} 
    \caption{Panel a: empirical correlation matrix for a subset of 100 genes and correlation matrices estimated with BASIL and PLIER. The two middle panels show the BASIL posterior mean of the correlation matrix without and with zeroing out the entries for which the corresponding $95\%$ credible interval included 0. Panel b: gene network plot for BASIL including only statistically significant correlations and PLIER. Positive (negative) correlations are shown in red (blue). }
    \label{fig:appWB}
\end{center}
\end{figure}
 
We show the resulting gene network in Figure~\ref{fig:appWB}b. 
BASIL identifies two large clusters  that contain genes associated with neutrophil granules and activation. The first includes genes such as arginase-1 (ARG1), matrix metalloproteinase-9 (MMP9), and peptidoglycan recognition protein 1 (PGLYRP1), which are known markers of neutrophil function \citep{yang_et_al_20}. The second group includes neutrophil elastase (ELANE), myeloperoxidase (MPO), azurocidin (AZU1) and defensin $\alpha4$ (DEFA4), reflecting the co-expression of azurophilic granule proteins in neutrophils \citep{calzetti_et_al_22}. These genes have previously been identified as a coordinated module under different conditions \citep{hemmat_et_al_20}. 
There is also a third cluster with genes specific to erythrocytes and erythropoiesis, such as erythroid 5'-aminolevulinate synthase (ALAS2), glycophorin B (GYPB) and erythrocyte membrane anion exchanger (SLC4A1). These genes are well known to be co-expressed in leukocyte-derived data reflecting erythropoietic adaptation \citep{Trudel2024}.
While BASIL yields a sparse network with three well-separated,  biologically interpretable clusters, PLIER's network  is considerably denser, with numerous weak or spurious correlations obscuring the underlying  modular structure.
This demonstrates BASIL's superior ability to separate signal from noise,  resulting in a sparse network that enhances biological interpretability and facilitates the 
identification of functionally relevant gene modules.

We refer to Section \ref{sec:oos} of the supplementary material for experiments that compare the out-of-sample predictive accuracy of competing methods, where BASIL displays state-of-the-art performance.

\section{Analysis of global fever data}

\subsection{Data description}

Infectious diseases remain a major cause of morbidity and mortality worldwide, and patients often present with clinically similar febrile syndromes for which rapid, etiology-specific diagnostics are unavailable. This uncertainty can lead to inappropriate antibiotic use, thereby accelerating antimicrobial resistance \citep{naghavi2022global}.  Motivated by the need for timely, pathogen-independent tools to support clinical decision-making and antibiotic stewardship, \cite{ko2023host} generated host whole-blood transcriptomic profiles that capture coordinated gene-expression programs induced by infection and inflammation. The resulting dataset comprises transcriptomic data from patients with confirmed viral or bacterial infections of varying etiologies, as well as from individuals with non-infectious disease mimics.
Because many of these responses are organized along known immunological pathways (e.g., interferon signaling, neutrophil activation, acute-phase response), incorporating curated gene sets can improve interpretability and stabilize inference in high-dimensional settings. This motivates analyzing the global fever cohort using a pathway-informed latent factor model that captures shared biological modules underlying variation in whole-blood expression while accommodating substantial heterogeneity across pathogens, clinical presentations, and study sites. The dataset is publicly available through the Gene Expression Omnibus (GEO) under accession number GSE211567 \citep{ko2023host}.

The global fever cohort included 101 participants with bacterial infections (41 bloodstream infections and 60 bacterial zoonoses), 123 with viral infections (80 respiratory and 43 dengue), and 67 with non-infectious illnesses (e.g., pulmonary embolism, congestive heart failure, COPD/asthma, cancer, and autoimmune disorders) \citep{ko2023host}. The study included 291 patients recruited in the United States (150) and Sri Lanka (141) and ensured a balanced representation across demographic and clinical subgroups.  Pathogens stratified by country are reported in Table \ref{tab:pathogen_country_counts} in the supplementary material.
RNA was extracted from whole-blood samples and sequenced to measure gene expression in two batches, with batch effects corrected for by using overlapping samples. The authors filtered out low expression genes and normalized the data, and the resulting dataset included 14,386 genes.

To construct a gene set matrix $C$, we first mapped RefSeq transcript identifiers from the GSE211567 normalized expression data to Ensembl gene identifiers using the org.Hs.eg.db annotation database. We retrieved Gene Ontology Molecular Function (GO:MF) gene sets from the Molecular Signatures Database (MSigDB) for Homo sapiens (collection C5) and merged these annotations with our identifier mappings. We then generated a binary gene sets matrix where rows represent genes and columns represent gene sets, with entries indicating gene membership (1) or non-membership (0) in each set. To ensure robust downstream analysis, we filtered out pathways with fewer than 10 genes and genes that were not annotated to any gene set.

\subsection{Results}
We apply BASIL to the global fever data.
The criterion introduced in \eqref{eq:JIC} selects 19 factors and the ratio between $\tau_{\Gamma}^2$ and $\tau_{\Psi}^2$ is 8.2, suggesting that gene ontology contributes substantially in modeling the gene covariance. Figure~\ref{fig:method} compares the correlation matrices estimated by BASIL and PLIER with the empirical correlation matrix.  BASIL's estimates are more accurate, particularly for the negative coefficients, providing more robust and reliable results.

The structure introduced by BASIL on the loadings matrix enables direct interpretation in terms of gene ontology. Specifically, the graphical matrix $\Gamma$ links latent factors to biological pathways, allowing functional explanation of factors using established gene ontology terms. To focus on statistically significant associations, we applied the BASIL uncertainty quantification procedure  and set all entries in $\Gamma$ to zero whose $95\%$ credible intervals included zero (about $40\%$ of the matrix entries). 
Our analysis concentrates on the first ten factors, representing the components that explain the largest proportion of variability in the data. The results are shown in Figure~\ref{fig:appGF}a, displaying for each factor the five pathways with the largest absolute coefficients.

The inferred factors define interpretable axes of molecular function that delineate signaling, transport, and immune regulatory pathways. The most prominent signal is phosphoinositide phosphatase activity: three factors (Factors 1, 2 and 7) are dominated by gene ontology terms for dephosphorylation of phosphatidylinositol derivatives, indicating that variation in phosphoinositide turnover is a major driver of expression heterogeneity. 
Beyond these core signaling modules, several factors capture structured trade-offs between functional programs. Factor 4 showed an innate-immune signature, linking high ATPase-coupled ion transport with low Toll-like receptor 4 (TLR4) binding activity, which mediates recognition of lipopolysaccharide (LPS) and triggers inflammatory responses \citep{luo2025examination, park2013recognition}.
This immune-versus-transport contrast likely reflects cellular heterogeneity between inflammatory and non-immune cell populations in blood. This factor exemplifies the critical value of allowing both positive and negative loadings to capture biologically meaningful functional oppositions. Similarly, Factor 2 captured an opposing relationship between membrane transporters (including ABC-type glutathione S-conjugate transporters) and phosphoinositide phosphatases, while Factor 9 contrasts phosphoinositide metabolism with macromolecular transport. Such bidirectional patterns, which reflect the coordinated activation and suppression of complementary processes, would be obscured in approaches that constrain all loadings to be positive, such as PLIER. Additional factors captured more specialized cellular functions. Factor 8 was enriched for membrane-associated processes, including binding of cardiolipin and phosphatidylglycerol - lipids enriched in mitochondrial and bacterial membranes - coupled with ion channel regulation and transporter activity, potentially reflecting mitochondrial stress responses or membrane remodeling. Factors 3, 5, and 10 lack a clear unifying theme, as their top gene ontology terms span disparate functions.

\begin{figure}[t!]
\begin{center}
 \includegraphics[width=0.74\textwidth]{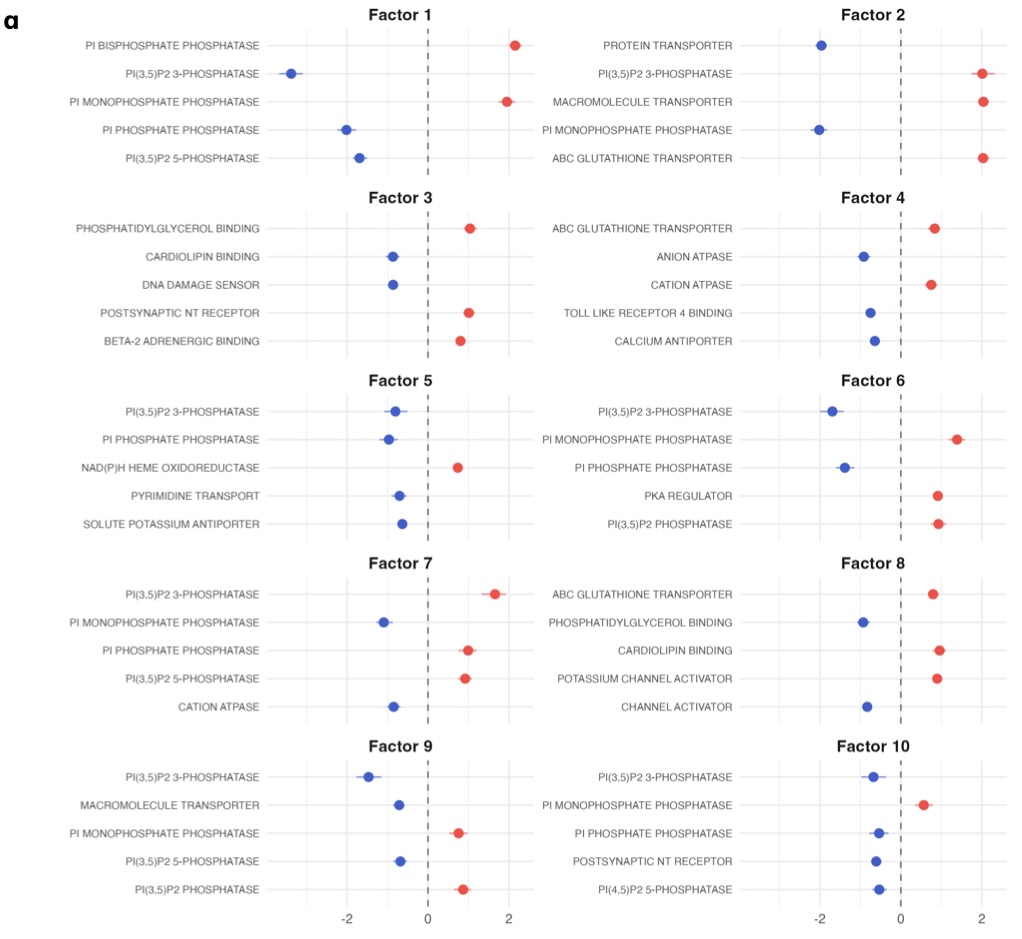} 
    \includegraphics[width=0.74\textwidth]{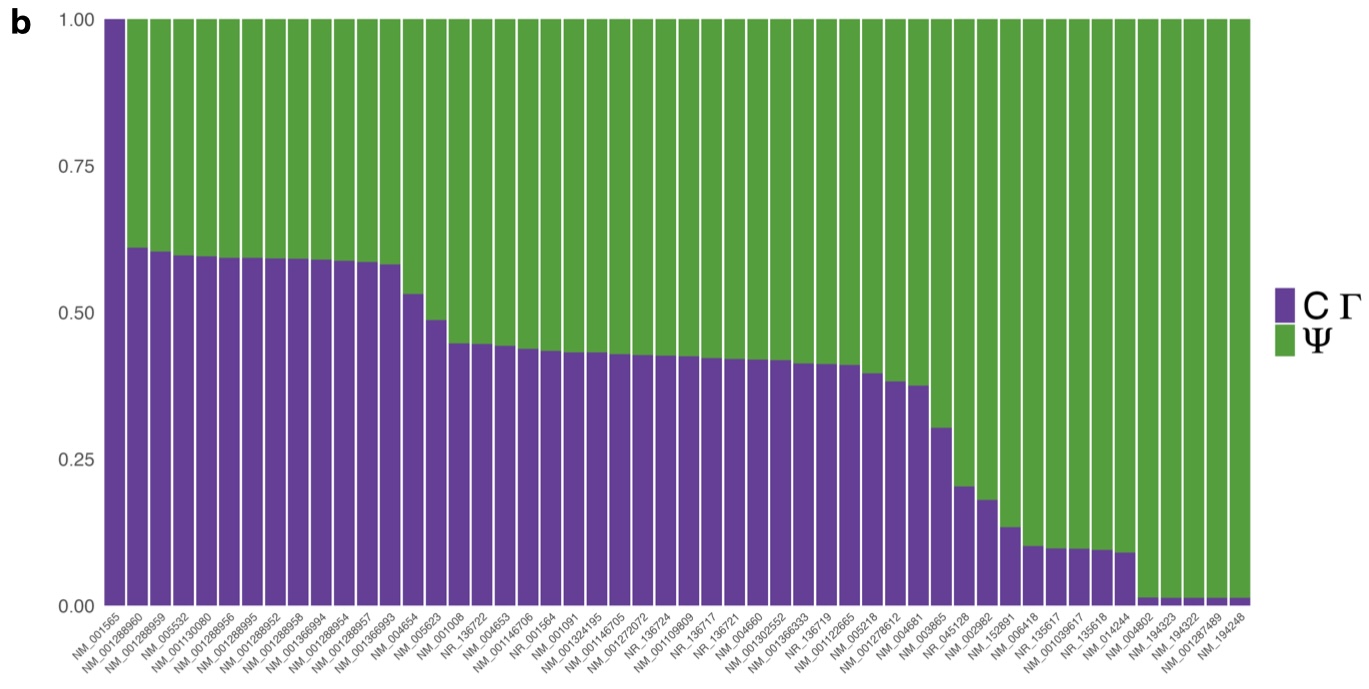} 
    \caption{Panel a: top 5 pathways for the first ten factors of $\Gamma$ after zeroing out the entries for which the corresponding 95\% credible intervals included 0. The plot shows the posterior mean and 95\% posterior credible interval for each element. Positive (negative) values are shown in red (blue). Panel b: proportions of variance explained by known ($C\Gamma$) versus unknown ($\Psi$) pathways for a subset of 50 genes with the highest variability.}     
    \label{fig:appGF}
\end{center}
\end{figure}

Figure~\ref{fig:appGF}b shows what fractions of systematic variability are explained by known pathways $(C \Gamma)$ or unknown components ($\Psi$) for the 50 genes with the highest variability. The variability of the gene C-X-C motif chemokine 10 (NM$-001565$ in the figure) is entirely explained by known pathways. This gene encodes a potent pro-inflammatory chemokine induced by IFN$-\gamma$ and is part of a well-characterized immune pathway. For instance, CXCL10 is included in the KEGG Chemokine Signaling pathway and the Hallmark Inflammatory Response gene set \citep{liu_11, ko2023host}. Hence, when interferon or inflammation pathways are active, CXCL10 is coordinately upregulated. Similarly, we note several IFI27 transcripts with more than half of their variation explained by known pathways. IFI27 is upregulated by IFN$-\alpha/\beta$ during antiviral responses \citep{Shojaei_25}. The dominance of CXCL10 and IFI27 indicates that type I/II interferon-driven inflammation is a major source of structured expression variance in the dataset. These genes are markers of an active immune response to infection or inflammatory challenge. Therefore, the top genes confirm that canonical immune pathways (especially interferon signaling) are a predominant, explainable source of gene expression variation in this dataset. Among the genes with the lowest fraction of variability explained by known pathways, we note some Otoferlin (OTOF) transcripts. OTOF acts as a calcium sensor for synaptic vesicle fusion in auditory inner hair cells \citep{otof} but is not expressed in immune or inflammatory cells.

\section{Discussion}

We have presented BASIL, a Bayesian matrix factorization framework for analyzing bulk RNA-sequencing data that integrates gene set annotations to extract biologically interpretable latent factors. Through comprehensive simulation studies and 'omics data analysis, we demonstrated that BASIL offers substantial improvements over existing methods in several critical aspects: accurate recovery of gene covariance matrix, rigorous uncertainty quantification, and enhanced interpretability of latent components.
A key advantage of BASIL is its fully probabilistic framework, which avoids the ad hoc post-processing required by methods like PLIER and naturally provides uncertainty quantification without the instability and computational expense of cross-validation approaches. BASIL is computationally efficient, making it broadly applicable without extensive tuning or computational resources.

While our proposed framework addresses several limitations in the current literature, there remain many directions for further development.
First, while we focused on bulk RNA-seq data, the framework could be extended to single-cell RNA-sequencing, where additional structure (e.g., cell type annotations, temporal dynamics) could be incorporated. Other interesting directions include  multi-study or multi-view settings, that are specifically designed for integrating diverse data modalities such as multi-omics. One promising direction would be extending scalable approaches in those settings \citep{blast, fama} to incorporate gene ontologies. Finally, it would be interesting to extend our methodology to model clustered data, letting the factor loadings to vary across clusters. This would allow the grouping of patients into biologically interpretable clusters based on their RNA-sequencing data, with potential applications to precision medicine and individualized treatment development.

\section*{Acknowledgments}
This work was partially supported by the National Institute of Health (NIH) (grants R01-ES035625, R01-ES027498, P42-ES010356, R01-AI155733), the Office of Naval Research (ONR) (grant N000142412626), the European Research Council under the European Union’s Horizon 2020 research and innovation (grant agreement No 856506), and the National Science Foundation (NSF) (IIS-2426762).

\section*{Disclosure statement}\label{disclosure-statement}

The authors have no conflicts of interest to declare 

\section*{Data Availability Statement}\label{data-availability-statement}

Both RNA-seq datasets used in this study are publicly available. The validation dataset provided in \cite{mao2019pathway} is available in the \texttt{PLIER} R package, and the global fever dataset was downloaded from the GEO (accession nos. GSE211567). 

\phantomsection\label{supplementary-material}
\bigskip

\begin{center}

{\large\bf SUPPLEMENTARY MATERIAL}

\end{center}

\begin{description}
\item[Supplementary Information:]
Additional details, results and plots and proof of theoretical results.
\item[R-package for BASIL:]
R-package containing code to run BASIL, together with a vignette illustrating how to use it. It is available at \url{https://github.com/federicastolf/BASIL}.
\item[Global fever analysis and simulations:]
Code to reproduce the simulations and the analysis of global fever data, together with the pre-processing of the dataset. It is available at \url{https://github.com/federicastolf/BASIL_paper}.
\end{description}

\bibliography{bibliography}


\newpage

\section*{Supplementary Material for \say{Pathway-based Bayesian factor models for 'omics data}}

\setcounter{section}{0}
\setcounter{table}{0}
\setcounter{figure}{0}
\setcounter{equation}{0}
\renewcommand{\thefigure}{S\arabic{figure}}
\renewcommand{\thesection}{S.\arabic{section}}
\renewcommand{\thetable}{S\arabic{table}}
\renewcommand{\theequation}{S.\arabic{equation}}

\section{BASIL estimates}

\begin{algorithm}[H]
\caption{BASIL: Bayesian Analysis with gene-Sets Informed Latent space}
\label{alg:basil}
\begin{algorithmic}[1]

\Require Standardized gene expression matrix $Y\in\mathbb{R}^{n \times p}$,
gene-set matrix $C\in\{0,1\}^{p\times q}$, upper-bound to latent dimension $k_{\max}$

\State \textbf{Estimate $k$.}
For each $k = 1, \dots, k_{\max}$ compute
\[
\mathrm{JIC}(k)= -2\ell_k + k\max(n,p)\log\{\min(n,p)\}
\]
Set $\hat k=\arg\min_{k\in\{1, \dots, k_{\max}\}}\mathrm{JIC}(k)$.

\State \textbf{Orthogonal decomposition.}
Compute orthonormal bases
$B_{\mathcal C}$ for $\mathcal{C}(C)$
and $B_{\mathcal C}$ for $\mathcal{N}(C)$.

\State \textbf{Pre-estimate the latent factors.}
Compute rank-$\hat k$ SVD
$Y\approx UDV^\top$,
set $\widehat{M}=\sqrt{n}\,U$.

\State \textbf{Separate source of variation in the data.}
Form
$Y_{\mathcal C}=YB_{\mathcal C}$,
$Y_{\mathcal N}=YB_{\mathcal N}$.

\State \textbf{Estimate shrinkage hyperparameters.}
Estimate 
$\widehat{\tau}_\Gamma^2$ and $\widehat{\tau}_\Psi^2$ as describe in Section \ref{subsec:hyperparameters}.

\State \textbf{Point estimates.} Obtain point estimates for the loadings via their posterior mean, which is available in closed form in \eqref{eq:posterior_mean_lambda}--\eqref{eq:posterior_mean_Lambda_outer}.

\State \textbf{Posterior updates.}
Draw posterior samples for $\Lambda_{\mathcal C}, \Lambda_{\mathcal N}$, and $\sigma^2$ via \eqref{eq:posterior_params} and \eqref{eq:posterior_update_cc}.

\State \textbf{Posterior updates for $\Lambda$.} Obtain the induced posterior samples for $\Lambda$ via \eqref{eq:loads2}, and, optionally, for $\Gamma$ and $\Psi$ via $\Gamma = (C^\top C)^{-1} C^\top B_{\mathcal C} \Lambda_{\mathcal C}$ and $\Psi = B_{\mathcal N}\Lambda_{\mathcal N}$

\Statex \hspace{-2em} \textbf{Output:} Point estimates and posterior samples for $\Lambda$ (and, optionally, for $\Gamma$ and $\Psi$) and $\sigma^2$.

\end{algorithmic}

\end{algorithm}

\section{Proofs of theoretical results}
We start by introducing some additional notation. For a matrix $A$, we denote by $||A||$, $||A||_2$ its spectral and Frobenius norm, respectively, , and by $tr(A)$ its trace. 
\subsection*{Proof of the main result}

First, we propose the following point estimator for the idiosyncratic variance, $\hat \sigma^2 = ||(I_n - U U^{\top}) Y ||_2^2/[(n - k) p]$, and show its consistency. 
\begin{proposition}\label{prop:sigma_hat_consistency}
    Under Assumptions \ref{assumption:data}--\ref{assumption:hyperparameters}, as $n \to \infty $, with probability at least $1-o(1)$, we have
    \begin{equation*}
       \hat \sigma^2 = \sigma_0^2 + r_n,
    \end{equation*}
    where $|r_n| \lesssim \frac{1}{\sqrt{n}} + \frac{1}{\sqrt{p_n}}$.
\end{proposition}

\begin{proof}[Proof of Proposition \ref{prop:sigma_hat_consistency}]
    First, consider the following decomposition 
    \begin{equation*}
        \begin{aligned}
            (I_n - U U^{\top}) Y = (I_n - U_0 U_0^{\top}) Y + (U_0 U_0^{\top} - U U^{\top}) Y  = (I_n - U_0 U_0^{\top}) E + (U_0 U_0^{\top} - U U^{\top}) Y.
        \end{aligned}
    \end{equation*}
Thus,
  \begin{equation*}
        \begin{aligned}
           ||(I_n - U U^{\top}) Y ||_2^2 & =  ||(I_n - U_0 U_0^{\top}) E||_2^2  +|| (U_0 U_0^{\top} - U U^{\top}) Y||_2^2 + 2 tr(E^\top(I_n - U_0 U_0^{\top}) (U_0 U_0^{\top} - U U^{\top}) Y ).
        \end{aligned}
    \end{equation*}
    Next, we analyze each term separately.
    \begin{enumerate}
        \item $||(I_n - U_0 U_0^{\top}) E||_2^2 \sim \sigma_0 \chi_{(n-k)p}^2$, where $ \chi_{\nu}^2$ denotes the $\chi$-squared distribution with $\nu$ degrees of freedom.
        \item With probability at least $1-o(1)$, $|| (U_0 U_0^{\top} - U U^{\top}) Y||_2 \lesssim ||U_0 U_0^{\top} - U U^{\top}|| ||Y||_2 \lesssim \frac{\sqrt{p}}{\sqrt{n}} + \frac{\sqrt{n}}{\sqrt{p}}$, since, with probability at least $1-o(1)$, 
        $ ||U U^\top - U_0 U_0^\top || \lesssim \frac{1}{n} + \frac{1}{p}$ by Proposition \ref{prop:reovery_P} and $||Y||_2 \lesssim \sqrt{np}$ by Proposition \ref{prop:norm_Y}.
        \item With probability at least $1-o(1)$, $| tr(E^\top(I_n - U_0 U_0^{\top}) (U_0 U_0^{\top} - U U^{\top}) Y )| \leq ||E||_2 || (U_0 U_0^{\top} - U U^{\top}) Y||_2 \lesssim \sqrt{n} + \sqrt{p} + \frac{p}{\sqrt{n}} + \frac{n}{\sqrt{p}}$, since, with probability at least $1-o(1)$, 
        $ || (U_0 U_0^{\top} - U U^{\top}) Y||_2 \frac{\sqrt{p}}{\sqrt{n}} + \frac{\sqrt{n}}{\sqrt{p}}$ by the previous point and $||E||_2 \lesssim \sqrt{np}$ by Proposition \ref{prop:norm_E}.       
    \end{enumerate}
Moreover, we have $ \bigg|\frac{||(I_n - U_0 U_0^{\top}) E||_2^2}{(n-k)p} - \sigma_0^2\bigg| \lesssim \frac{1}{\sqrt{n}} + \frac{1}{\sqrt{p}}$, with probability at least $1-o(1)$, by Corollary 2.8.3 of \citet{vershynin_book}.
Combining all of the above and letting $r_n = \hat \sigma^2 - \sigma_0^2$, we have
\begin{equation*}
    \begin{aligned}
        |r_n| \lesssim \frac{1}{\sqrt{n}} + \frac{1}{\sqrt{p}} ,+ \frac{1}{(n-k)p} \big(\frac{p}{n} + \frac{n}{p} + \sqrt{n} + \sqrt{p} + \frac{p}{\sqrt{n}} + \frac{n}{\sqrt{p}}\big) \lesssim   \frac{1}{\sqrt{n}} + \frac{1}{\sqrt{p}},
    \end{aligned}
\end{equation*}
with probability at least $1-o(1)$.    
\end{proof}

\begin{proof}[Proof of Theorem \ref{thm:prior_variances}]
    It follows as a corollary from Propositions \ref{prop:sigma_hat_consistency} and \ref{prop:L_consistency}.
\end{proof}

\subsection*{Auxiliary Results}

Throughout, we denote by $M_0$ the true latent factors.

\begin{proposition}[Proposition 3.5 of \citet{fable}]\label{prop:reovery_P}
    Let $U$ be the matrix of left singular vectors associated to the leading $k$ singular values of $Y$ and $U_{0}$ be the matrix of left singular vectors of $M_0 \Lambda_0^\top$. Then, under Assumptions \ref{assumption:data}--\ref{assumption:hyperparameters}, as $n \to \infty $, with probability at least $1-o(1)$, we have
    \begin{equation*}
        ||U U^\top - U_0 U_0^\top || \lesssim \frac{1}{n} + \frac{1}{p}.
    \end{equation*}
\end{proposition}

\begin{proposition}\label{prop:norm_E}
    Under Assumption \ref{assumption:data}--\ref{assumption:Lambda}, we have 
    \begin{equation*}
    \begin{aligned}
     ||EP_C|| &\lesssim \sqrt{n} +\sqrt{q},& \quad ||E(I_p - P_C)|| &\lesssim \sqrt{n} + \sqrt{p-q},\\
          ||EP_C||_2 &\lesssim \sqrt{nq},& \quad ||E(I_p - P_C)||_2 &\lesssim \sqrt{n(p-q)}, \\
           ||U_0^\top EP_C||_2 &\lesssim \sqrt{kq},& \quad ||U_0^\top E(I_p - P_C)||_2 &\lesssim \sqrt{k(p-q)},  
    \end{aligned}   
    \end{equation*}
    with probability at least $1-o(1)$.
\end{proposition}
\begin{proof}[Proof of Proposition \ref{prop:norm_E}]
    The results follow from Corollary 5.35 of \citet{vershynin_12} after noting that $||EP_C||_2 = ||E B_{\mathcal C}||_2$, $||E(I_p - P_C)||_2 = ||E B_{\mathcal N}||_2$, $||U_0^\top EP_C||_2 = ||U_0^\top E B_{\mathcal C}||_2$, and $||U_0^\top E(I_p - P_C)||_2 = ||U_0^\top E B_{\mathcal N}||_2$, with $E B_{\mathcal C} \in \mathbb R^{n \times q}$, $E B_{\mathcal N} \in \mathbb R^{n \times (p-q)}$, $U_0^\top E B_{\mathcal C} \in \mathbb R^{k \times q}$, and $U_0^\top E B_{\mathcal N} \in \mathbb R^{k \times (p-q)}$ having entries being distributed as independent Gaussian random variables with $0$ mean and $\sigma_0$ standard deviation.
\end{proof}

\begin{proposition}\label{prop:norm_Y}
    Under Assumption \ref{assumption:data}--\ref{assumption:Lambda}, we have 
    \begin{equation*}
        ||YP_C||_2 \lesssim \sqrt{n}||C\Gamma_0||_2 + \sqrt{nq}, \quad ||Y(I - P_C)||_2 \lesssim \sqrt{n}||\Psi_0||_2 + \sqrt{n(p-q)}, 
    \end{equation*}
    with probability $1-o(1)$.
\end{proposition}
\begin{proof}[Proof of Proposition \ref{prop:norm_Y}]
    First, note $YP_C = M_0 \Gamma_0^\top C^\top + E P_C$. Next, note that, with probability at least $1-o(1)$, $||M_0 \Gamma_0^\top C^\top  ||_2 \leq ||M_0|| ||\Gamma_0^\top C^\top  ||_2 \lesssim (\sqrt{n} + \sqrt{k})||\Gamma_0^\top C^\top  ||_2  \asymp \sqrt{n}||\Gamma_0^\top C^\top  ||_2$, since $||M_0|| \lesssim \sqrt{n} + \sqrt{k}$, with probability at least $1-o(1)$, by Corollary 5.35 of \citet{vershynin_12}, proving the first result in combination with Proposition \ref{prop:norm_E}. The second result is proven with analogous steps.
\end{proof}

Then, consider the two following quantities, $L_{\mathcal C}= || P_C V D||_F^2 /n$, $L_{\mathcal N}= || (I_p - P_C) V D||_F^2 /n$, measuring the magnitude of the signal explained by known and unknown pathways respectively. 
\begin{proposition}\label{prop:L_consistency}
    Under Assumptions \ref{assumption:data}--\ref{assumption:hyperparameters}, as $n \to \infty $, we have
    \begin{equation*}
       L_{\mathcal C} = ||C \Gamma_0||_2^2 + r_{\mathcal C, n}, \quad   L_{\mathcal N} = ||\Psi_0||_2^2 + r_{\mathcal N, n}, 
    \end{equation*}
where $|r_{\mathcal C, n}| \lesssim ||C \Gamma_0||_2^2 (1/\sqrt{n}  + 1/p) +||C \Gamma_0||_2(1+ \sqrt{q/n} ) + q/n + q/p$ and $|r_{\mathcal N, n}| \lesssim ||\Psi_0||_2^2 (1/\sqrt{n} + 1/p) +  ||\Psi_0||_2(1 + \sqrt{(p-q) / n} ) + (p-q) / n +(p-q) / p$, with probability at least $1-o(1)$.
\end{proposition}
\begin{proof}[Proof of Proposition \ref{prop:L_consistency}]
    We only show the result for $ L_{\mathcal C}$, since the one for $ L_{\mathcal N}$ follows with similar steps.     
    Since $P_C VD = P_C Y^\top U$, we have the following decomposition, 
    \begin{equation*}
        \begin{aligned}
          P_C VD^2 V^\top P_C =& P_C Y^\top U U^\top Y P_C   = P_C Y^\top U_0 U_0^\top Y P_C + P_C Y^\top (U U^\top -  U_0 U_0^\top) Y P_C   \\
          =& C \Gamma_0 M_0^\top M_0 \Gamma_0^\top C^\top + C \Gamma_0 M_0^\top E P_C + P_C E^\top  M_0 \Gamma_0^\top C^\top + P_C E^\top U_0 U_0^\top E P_C \\
          &+ P_C Y^\top (U U^\top -  U_0 U_0^\top) Y P_C\\
          =& n C \Gamma_0 \Gamma_0^\top C^\top +  C \Gamma_0 (M_0^\top M_0 - nI_k)\Gamma_0^\top C^\top+ C \Gamma_0 M_0^\top E P_C + P_C E^\top  M_0 \Gamma_0^\top C^\top \\
          &+  P_C E^\top U_0 U_0^\top E P_C + P_C Y^\top (U U^\top -  U_0 U_0^\top) Y P_C\\ 
        \end{aligned}
    \end{equation*}
    Next, we bound each term.
    \begin{enumerate}
        \item With probability at least $1-o(1)$, we have
        $ ||C \Gamma_0 (M_0^\top M_0 - nI_k)\Gamma_0^\top C^\top ||_2 \leq ||C \Gamma_0||_2^2 || M_0^\top M_0 - nI_k|| \lesssim ||C \Gamma_0 ||_2^2 (\sqrt{n} + \sqrt{k})$, since $|| M_0^\top M_0 - nI_k||\lesssim \sqrt{n} + \sqrt{k}$ by Corollary 5.35 of \citet{vershynin_12}. 
        \item With probability at least $1-o(1)$, we have 
        \begin{equation*}
            \begin{aligned}
                ||C \Gamma_0 M_0^\top E P_C||_2 \leq ||C \Gamma_0 ||_2 ||M_0|| ||E P_C||  \lesssim ||C \Gamma_0 ||_2 (\sqrt{n} + \sqrt{k}) (\sqrt{n} + \sqrt{q}) \lesssim ||C \Gamma_0 ||_2 (n + \sqrt{n q}). 
            \end{aligned}
        \end{equation*}
        since, with probability at least $1-o(1)$, $||EP_C|| \lesssim \sqrt{n} + \sqrt{q}$ by Proposition \ref{prop:norm_E}, and $||M_0|| \lesssim \sqrt{n} + \sqrt{k}$ by Corollary 5.35 of \citet{vershynin_12} 
        \item With probability at least $1-o(1)$, we have 
        \begin{equation*}
            \begin{aligned}
||P_C E^\top U_0 U_0^\top E P_C||_2 \leq ||P_C E^\top U_0||_2^2 \lesssim \sqrt{kq}
            \end{aligned}
        \end{equation*}
        by Proposition \ref{prop:norm_E}.

        \item With probability at least $1-o(1)$, we have 
        \begin{equation*}
            \begin{aligned}
                || P_C Y^\top (U U^\top -  U_0 U_0^\top) Y P_C||_2 &\leq ||YP_C||_2^2 ||U U^\top -  U_0 U_0^\top||  \lesssim (\sqrt{n}||C \Gamma||_2 + \sqrt{nq} )^2\big(\frac{1}{n} + \frac{1}{p}\big) \\
                &= (||C \Gamma||_2^2 + q)\big(1 + \frac{n}{p}\big)
            \end{aligned}
        \end{equation*}
since, with probability at least $1-o(1)$, 
        $ ||U U^\top - U_0 U_0^\top || \lesssim \frac{1}{n} + \frac{1}{p}$ by Proposition \ref{prop:reovery_P} and $||Y P_C||_2 \lesssim \sqrt{n}||C \Gamma|| + \sqrt{nq}$ by Proposition \ref{prop:norm_Y}.
    \end{enumerate}
Combining all the above and letting $r_{\mathcal C, n} = L_{\mathcal C} - ||C \Gamma_0||_2^2$, we obtain
\begin{equation*}
    \begin{aligned}
        |r_{\mathcal C, n} | \lesssim ||C \Gamma_0||_2^2\big(\frac{1}{\sqrt{n}} + \frac{1}{p}\big) +  ||C \Gamma_0||_2\big(1 + \sqrt{\frac{q}{n}} \big) + \frac{q}{n} + \frac{q}{p}
    \end{aligned}
\end{equation*}
with probability at least $1-o(1)$.
Analogous steps lead to $ |r_{\mathcal N, n} | \lesssim ||\Psi_0||_2^2 (1/\sqrt{n} + 1/p) +  ||\Psi_0||_2\big(1 + \sqrt{(p-q) / n} \big) + (p-q) / n +(p-q) / p$, with $r_{\mathcal N, n} = L_{\mathcal N} - ||\Psi_0||_2^2$.
\end{proof}

\section{Additional details and results for global fever data }

\subsection{Additional details on global fever dataset}

\begin{table}[h]
\centering
\caption{Pathogen stratified by country in the global fever dataset.}
\label{tab:pathogen_country_counts}
\begin{tabular}{llrr}
\hline
Pathogen & Infection type & Sri Lanka & United States \\
\hline
Coxiella burnetii & Bacterial    & 3  & 0  \\
Dengue             & Viral        & 43 & 0  \\
Enterobacter       & Bacterial    & 0  & 17 \\
Influenza A and B     & Viral        & 38 & 29 \\
Leptospira         & Bacterial    & 30 & 0  \\
Non-infection       & Non-infection & 0  & 67 \\
Respiratory virus, other   & Viral        & 0  & 13 \\
Rickettsia         & Bacterial    & 27 & 0  \\
Staphylococcus     & Bacterial    & 0  & 10 \\
Streptococcus      & Bacterial    & 0  & 14 \\
\hline
\end{tabular}
\end{table}

\subsection{Additional gene set enrichment analysis for global fever data}

The BASIL factors were interpreted by applying gene set enrichment analysis (GSEA) to the loadings separately for each factor. While the Gene Ontology Molecular Function collection was used to construct the factors to reduce selection bias in gene set curation, generating contextualized interpretations using this collection is challenging. Therefore, the human Reactome pathways \citep{ligtenberg2025reactomedb} were used in the GSEA. 
The Ensemble identifiers were mapped to their Entrez counterparts using AnnotationDbi \citep{pages2025annotationdbi}, and Entrez identifiers with multiple loadings were summarized using the median. GSEA was carried out using fast GSEA \citep{korotkevich2016fast} within the clusterProfiler package \citep{xu2024using}. The Reactome pathways with 10-1,000 Entrez identifiers were used, and 100,000 permutations of the FGSEA-simple procedure were run to estimate p-values for those pathways. 
Pathways with a false discovery rate $\leq 5\%$ were used to construct network diagrams that could help to identify summary interpretations for each factor \citep{csardi2026igraph}. Pathways with a Jaccard index exceeding 0.4 were grouped, and then singleton and doublet groups were removed. Lastly, community detection was applied to generate homogenous groups for interpretation. The resulting clusters for the first two factors are showed in Figure \ref{fig:appEA}a.

\begin{figure}[t]
\begin{center}
 \includegraphics[width=0.99\textwidth]{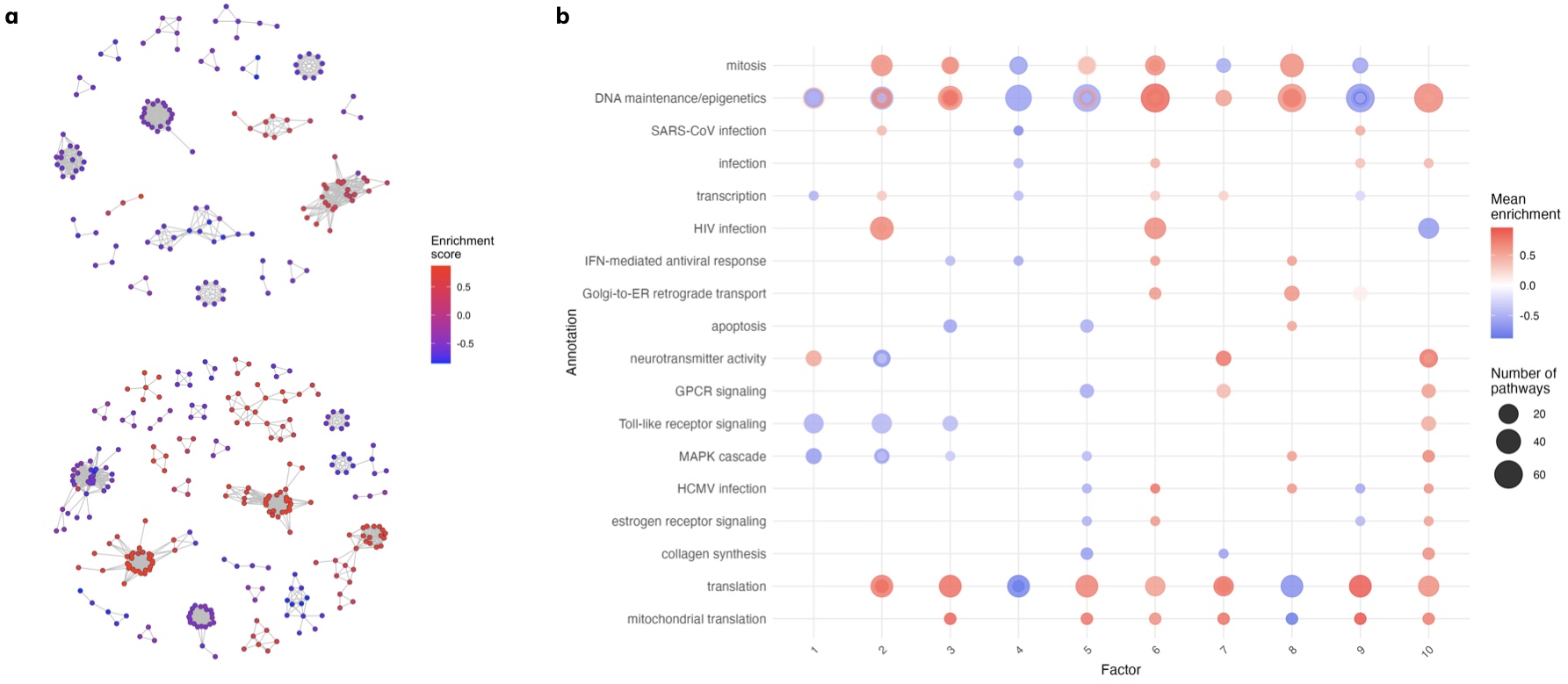} 
    \caption{Panel a: network maps for the first two factors, following filtering and community detection, overlaid with the enrichment scores. Panel b: factor annotations for the first ten factors. }     
    \label{fig:appEA}
\end{center}
\end{figure}

Figure \ref{fig:appEA}b displays the factor annotations for the first ten factors resulting from the GSEA.
In summary, the factor analysis identifies biologically interpretable axes of host immune variation, several of which align with established immune paradigms.
Figure \ref{fig:pairwiseGSEA} shows the pairwise plots for the elements of the first ten factors, colored by different pathogens. For example, Factor 2 is enriched for an HIV infection gene set and ribosomal/translation pathways, suggesting a host state prioritizing protein synthesis—consistent with responses to viral pathogens such as influenza or dengue that depend on host ribosomes for replication \citep{walsh2011viral}. 
Notably, some intracellular bacteria elicit similar interferon-dominated responses. Coxiella and Rickettsia infections, for instance, activate interferon-stimulated genes \citep{colonne2011rickettsia}, reflecting their intracellular lifestyle and reliance on host cellular machinery. 
Factor 10 is enriched for broad immune and stress-response pathways and uniquely shows positive enrichment for inflammatory signaling. 
The involvement of IL-1 and NF-$\kappa$B pathways is characteristic of bacterial infections, as pathogens such as Staphylococcus and Streptococcus strongly activate NF-$\kappa$B and inflammasome signaling \citep{liu2017nf}. Other factors capture variation in immune cell metabolic and proliferative states (Factors 3, 4, 8, and 9) or systemic stress responses (Factor 7).

\begin{figure}[t]
\begin{center}
    \includegraphics[width=0.97\textwidth]{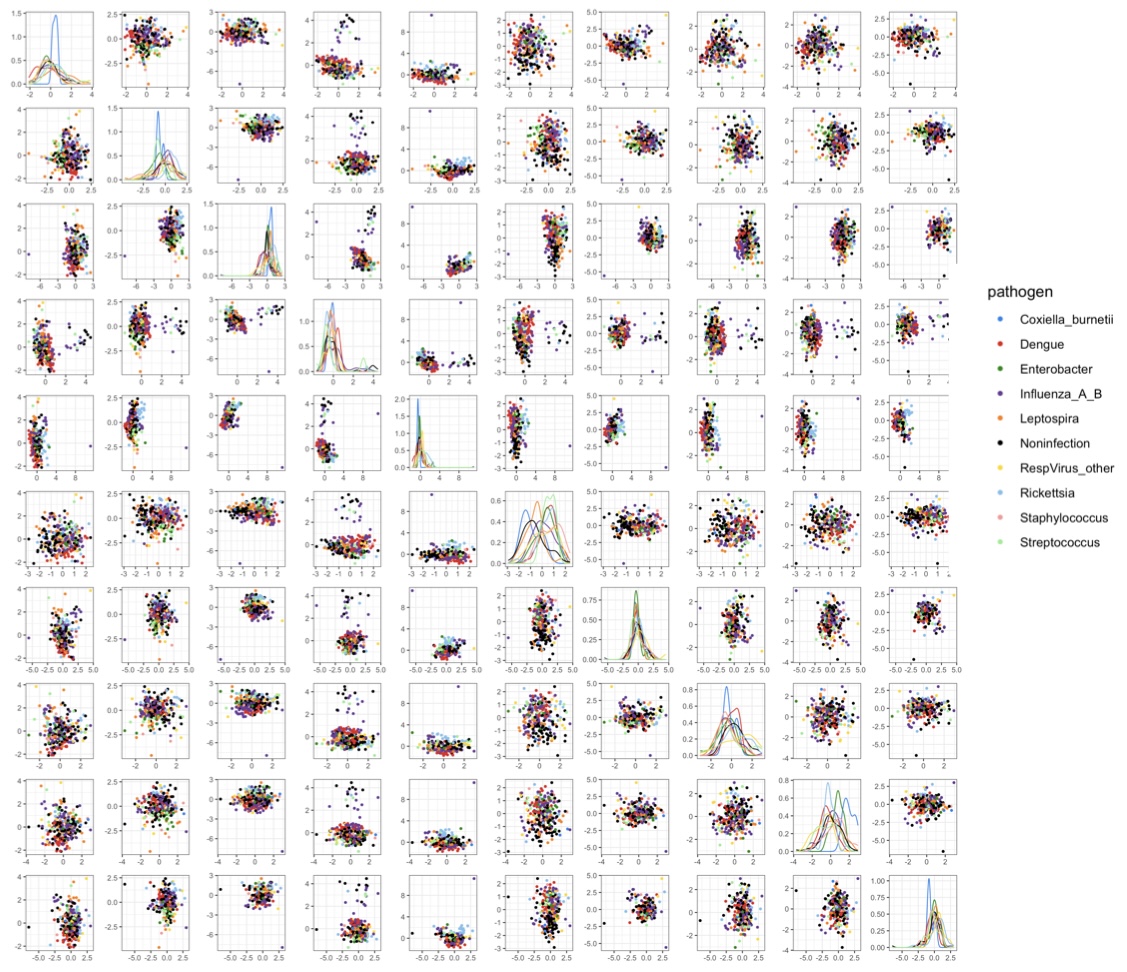} 
    \caption{Pairwise plots for the elements of $M$ are shown for the first ten factors. The sample-level points are colored by the pathogen. }
    \label{fig:pairwiseGSEA}
\end{center}
\end{figure}

\section{Out-of-sample predictive performance} \label{sec:oos}

We tested the out-of-sample predictive performance of competing methods in the two datasets analyzed (whole-blood RNA seq data and global fever data). For each data set, we perform 50 random $80\%-20\%$ train-test splits. We fit each method in the training set and compute the log-likelihood in the test set $\mathcal Y_{test}$, which, for a given estimator of the covariance $\hat \Sigma$, is given by
\begin{equation*}
    l(\hat \Sigma) = - \frac{n_{test}}{2}\log \big(2\pi|\hat \Sigma|\big) - \frac{1}{2} \sum_{y \in \mathcal Y_{test}} y^\top \hat \Sigma^{-1} y, 
\end{equation*}
where $n_{test}$ denotes the size of the test set. For BASIL, we used the posterior mean to estimate the covariance. For PLIER, we use the same adjustment described in the simulation section for the low-rank component and let $\hat \Sigma = \hat \Lambda (\hat M^\top \hat M) / n \hat \Lambda^{\top} + \hat \sigma^2 I_p$, where $\hat \sigma^2$ is the empirical variance of the residual variability in the training set. We set the maximum number of latent factors for the criterion in \eqref{eq:JIC} and ROTATE to the minimum number of principal components of the training data explaining at least $80\%$ of the variability. Both BASIL and PLIER are fitted using the latent dimension estimated via the criterion in \eqref{eq:JIC}. Figure~\ref{fig:oos} shows the out-of-sample log-likelihood across the 50 random splits. 

\begin{figure}[t]
\begin{center}
    \includegraphics[width=0.8\textwidth]{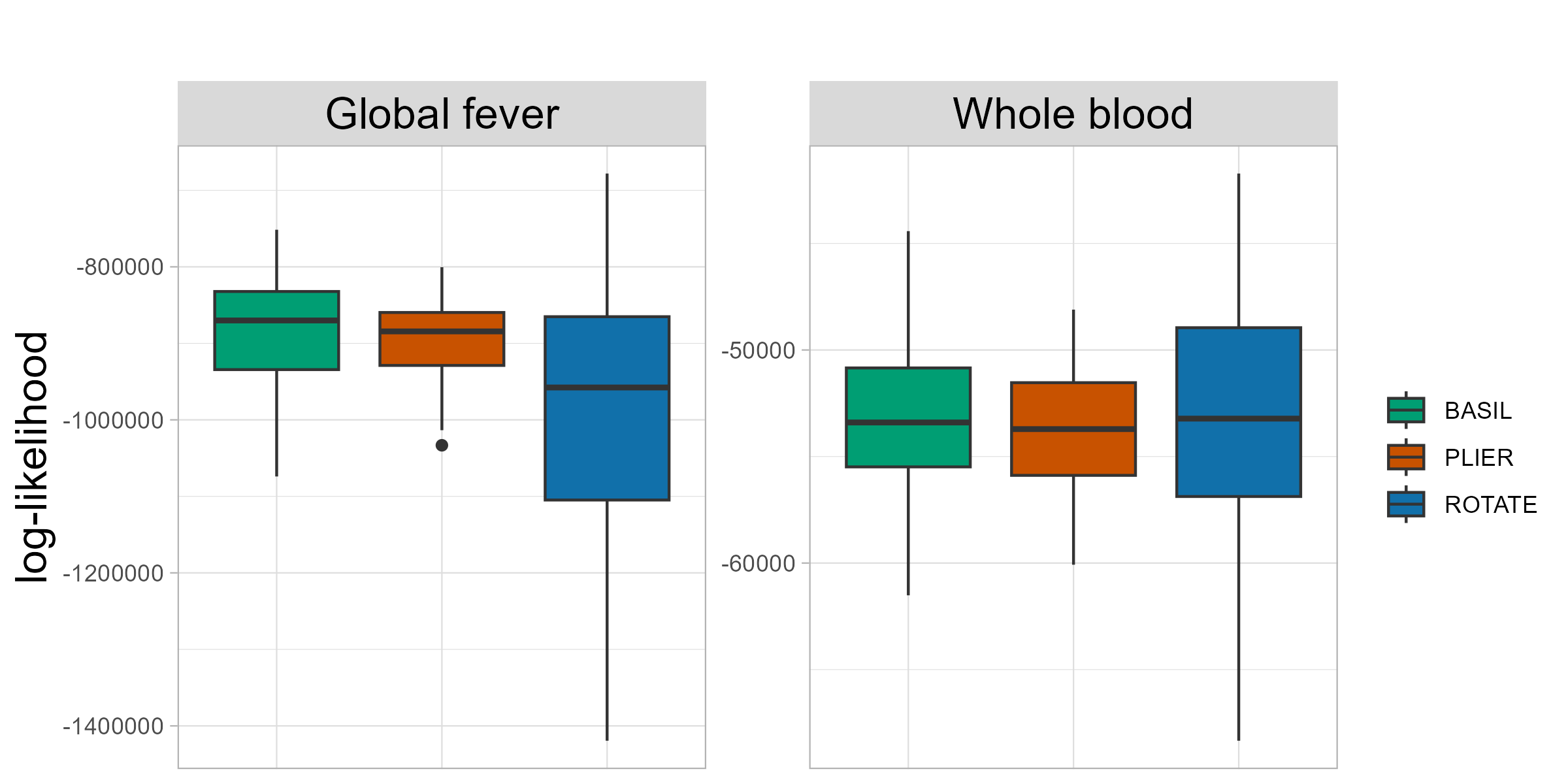} 
    \caption{ \textbf{Out-of-sample log-likelihood}. Out-of-log-likelihood of BASIL (green), PLIER (red), and ROTATE (blue) across 50 random $80\%-20\%$ train-test splits.}
    \label{fig:oos}
\end{center}
\end{figure}

In terms of predictive accuracy, the three methodologies perform comparably on the whole blood data set of \citet{mao2019pathway}, while, on the global fever data set of \citet{ko2023host}, ROTATE is the worst performing method, and BASIL slightly outperforms PLIER. 
Although the primary goal of BASIL is interpretable dimensionality reduction rather than prediction, out-of-sample predictive performance provides a critical assessment of model adequacy. Strong predictive performance indicates that the learned underlying representation  captures the biological signal in the gene co-expression patterns, rather than overfitting to training data. The comparable or improved performance of BASIL thus validates that our modeling framework is sufficiently flexible to accurately characterize the underlying structure of gene co-regulation.

\end{document}